\documentclass{article}

\usepackage{amsmath}
\usepackage{amssymb}
\usepackage{amsthm}
\usepackage{color}
\usepackage{bbm}
\usepackage{yfonts}
\usepackage{dsfont}
\usepackage{mathtools}
\usepackage{enumerate}
\usepackage{authblk}

\usepackage[round]{natbib}

\newcommand{\N }{\ensuremath{\mathbb N}}
\newcommand{\R }{\ensuremath{\mathbb R}}
\newcommand{\E }{\ensuremath{\mathbb E}}
\newcommand{\eR}{\ensuremath{\overline{\mathbb R}}}

\newcommand{\F } {\ensuremath{\mathbb{F}}}

\def\one{\mathds{1}}

\def\N{\mathds{N}}
\def\R{\mathds{R}}
\def\E{\mathds{E}}
\def\eR{\overline{\mathds{R}}}
\def\F{\mathds{F}}

\def\cB{\mathcal{B}}

\def\cT{\mathcal{T}}
\def\cI{\mathcal{I}}

\def\cH{\mathcal{H}}

\def\cP{\mathcal{P}}
\def\cL{\mathcal{L}}
\def\cF{\mathcal{F}}

\def\cG{\mathcal{G}}
\def\cN{\mathcal{N}}

\def\NA{\mathrm{NA}}

\def\d{\mathrm{d}}

\def\MSP{\textsc{Msp}}
\def\FTAP{\textsc{Ftap}}

\DeclareMathOperator{\clo}{CL}
\DeclareMathOperator{\epi}{epi}

\DeclareMathOperator{\dom}{dom}
\DeclareMathOperator{\supp}{supp}
\DeclareMathOperator{\cl }{cl}
\DeclareMathOperator{\inter}{int}
\DeclareMathOperator{\ri }{ri}
\DeclareMathOperator{\conv}{conv}

\DeclareMathOperator{\cone}{cone}
\DeclareMathOperator{\aff}{aff}

\DeclareMathOperator{\lin}{lin}

\newcommand{\close}[1]{\cl#1}%{\overline{#1}}

\newtheorem{theorem}{Theorem}[section]
\newtheorem*{theorem*}{Theorem}

\newtheorem{lemma}[theorem]{Lemma}

\newtheorem*{claim*}{Claim}

\newtheorem{corollary}[theorem]{Corollary}

\theoremstyle{definition}
\newtheorem{definition}[theorem]{Definition}

\newtheorem{assumption}[theorem]{Assumption}

\theoremstyle{remark}
\newtheorem{example}[theorem]{Example}
\newtheorem{remark}[theorem]{Remark}

\title{Robust martingale selection problem and its connections to the no-arbitrage theory}
\date{\small This version: \today}
\author{\textsc{Matteo Burzoni}\\
\vspace{-2.4mm}
{\normalsize Department of Mathematics, ETH Z\"urich\\
\texttt{matteo.burzoni@math.ethz.ch}}\\
\vspace{5mm}
\textsc{Mario \v{S}iki\'c}\\
\vspace{1.0mm}
{\normalsize Center for Finance and Insurance, Universit\"at Z\"urich\\
\texttt{mario.sikic@bf.uzh.ch}}
}

\begin{document}
\maketitle

\begin{abstract}
We analyze the martingale selection problem of \cite{Rok06} in a pointwise (robust) setting. We derive conditions for solvability of this problem and show how it is related to the classical no-arbitrage deliberations. We obtain versions of the Fundamental Theorem of Asset Pricing in examples spanning frictionless markets, models with proportional transaction costs and models for illiquid markets. In all these examples, we also incorporate trading constraints.
\end{abstract}
\textbf{Acknowledgments:} We would like to thank two anonymous referees for their insightful comments and Sven H\"{u}mbs for finding several typos in a previous version of the manuscript. Matteo Burzoni gratefully acknowledges support by the ETH Foundation.

\medskip\noindent
{\footnotesize{ \it Keywords:} Martingale Selection Problem, Arbitrage Theory, Fundamental Theorem of Asset Pricing, Markets with Frictions, Illiquidity.

\medskip\noindent
{\it  JEL subject classification: G12, C61.} 

\noindent
{\it AMS 2010 subject classification.} Primary 91B24, 60G42, 90C15  ; secondary 91G99, 90C39.
}

\section{Introduction}

In discrete-time, the martingale selection problem (\MSP) is stated as follows: Given two adapted families of random sets $V=(V_t)$ and $C=(C_t)$, find a family of pairs $(\xi,Q)$, consisting of an adapted process $\xi=(\xi_t)$ taking values in $V$ and a probability measure $Q$, such that 
$$
  \E_Q\big[\xi_{t+1}-\xi_t\,\big|\,\cF_t\big]\in C_t\qquad Q\mbox{-a.s.}
$$
The pair $(\xi,Q)$ is called a solution to the $\MSP$. When $C\equiv\{0\}$ the problem is asking for a sequence of selections $\xi$ of $V$ and a measure $Q$, such that $\xi$ is a $Q$-martingale. In the present form, the problem was first studied by \cite{Rok06}, where the measure $Q$ is also required to be equivalent to some chosen probability measure. We refer to this as the ``dominated setup''.

Martingale selection problems arise frequently in mathematical finance. The readers familiar with models of markets with frictions, described via solvency cones $K=(K_t)$, will recognize the \MSP\ with $V_t=\ri K_t^\ast$ and $C_t=\{0\}$ as precisely the dual formulation of absence of arbitrage. Indeed, in the literature, the pairs $(\xi,Q)$ are known as consistent price systems. What this observation is suggesting is that the possibility of solving a \MSP, namely ensuring the existence of sufficiently many pairs $(\xi,Q)$, can be related to a no-arbitrage condition for the associated financial market. 

In this paper we show that the connection goes much deeper and can be employed in a great variety of situations. We characterize the solvability of the \MSP\ with an approach closely related to that of~\cite{Rok}, albeit with some modifications. The idea is to identify a family of correspondences $W=(W_t)$, contained in $V$, which satisfies a certain dynamic programming principle. In particular, if we can solve the one-step \MSP\ for $(W_t,W_{t+1})$, then the general solution can be obtained by pasting together the one-step solutions. Since $W$ is contained in $V$, the resulting pairs are also solutions of the \MSP\ for $V$. Our main result is Theorem \ref{thm: main1} and is stated as follows.%; solvability means that for each $\omega\in\Omega$ there exists a solution $(\xi,Q)$ with $\omega\in\supp Q$.
\begin{theorem*}
The martingale selection problem $(V,C)$ is solvable if and only if $W_t(\omega)\neq\varnothing$ for all $t\in\cI$ and $\omega\in\Omega$.
\end{theorem*}
\noindent
%Solvability of a \MSP\ refers to existence of a rich enough family of selections of correspondences $V_t$, subject to additional properties; see Definition~\ref{def:solvability}.
The construction of correspondences $W$ is based on a backward iteration in the spirit of \cite{Rok}. However, the one considered in that paper is applicable only in the case when $V$ is open valued. We show how to suitably modify it; see Section~\ref{Section: MSP}, equation~\eqref{eq: iteration} and equation~\eqref{eq: Uflat}.

We work in a finite discrete-time pointwise setting. In particular, we do not assume the existence of any probabilistic description of the market and all the statements on random objects are meant to hold for any $\omega\in\Omega$. As a consequence, the set of probability measures $\cP$ for which we solve the \MSP\, is naturally chosen as the set of all finite support probability measures. This turns out to have significant advantages in terms of establishing measurability of the involved correspondences. 

\paragraph{On the fundamental theorem of asset pricing.}
In the dominated setup, a standard approach for showing the fundamental theorem of asset pricing (\FTAP) is functional analytic. The key step is to prove that the no-arbitrage condition implies that the set of superhedgeable claims is closed with respect to an appropriate topology. Once this is achieved, the \FTAP\ becomes a statement about the polar of this set being non-empty. This approach has been introduced in \cite{Sch94} in the case of a frictionless market model and it has been employed in the case of markets with proportional transaction costs; see \citep{K99,KRS02,KRS03,S04}. The most general formulation is that of currency markets initiated in \cite{K99} and based on the so-called solvency cones, whose main role is to determine the self-financing condition for trading strategies. Several different notions of arbitrage have been considered and corresponding version of \FTAP\ have been provided, possibly, under some additional technical conditions such as the ``efficient friction'' hypothesis (transaction costs are always non-trivial). We refer to \citep{KRS03} for an overview.

More recently, non-linear frictions, such as illiquidity effects, have been considered in the literature. Generalization of the currency market model to convex solvency region is given by \cite{AT07} on a finite probability space, and by \cite{PP10} on a general probability space. Both papers deal with claims with physical delivery. A continuous-time cash-delivery model is considered in \cite{CR12}, where price per unit for buying a certain number of shares of a risky asset is given by a so-called supply curve. In the same spirit, a general discrete-time model has been considered by \cite{P11} where a cost process, which depends on the traded volume, describes the cost of trading.

In the non-dominated setup the functional analytic approach is not universally applicable. The \FTAP\ needs to be obtained directly, using dynamic programming and measurable selection technology. An argument along those lines was first proposed in \cite{JS98} in the framework of frictionless markets. The same argument was successfully applied in the quasi-sure setup of \cite{BN15,BN16}. This setup has also been adopted in \cite{BZ17} for the case of a frictionless market with portfolio constraints and in \cite{BZ16} where, by using dynamic programming techniques close to the one used here, a \FTAP\ is shown for a market model with proportional transaction costs and under the efficient friction hypothesis. In a discrete-time setting with no probability measures, arbitrage theory has been investigated in \citep{ABPS16,BRS17,BFHMO16,CKT16,Ri15} for frictionless markets and in \citep{Bu16,BCKT17,DS14} for proportional transaction costs. In the latter group of papers some additional assumptions are taken: the existence of a cash account for the first one, constant transaction costs for the latter two. To the best of our knowledge a general theory for markets with frictions, in this pointwise setting, has not been established yet and it is addressed in this paper.

We tackle the problem as follows. As observed above, the \MSP\ can be interpreted as the dual problem to the existence of arbitrage strategies in models of financial markets. We ask whether the dual problem admits solutions; If not, we show how to construct an arbitrage strategy by convex analysis arguments. 

\paragraph{Overview of the paper.}
In Section 2 we define the setup and motivate it with examples from financial mathematics that we will treat in more detail in the last part of the paper. Section 3 is dedicated to the definition of certain projections of a correspondence and to the construction of the relevant objects for solving the $\MSP$. In Section 4 we define formally the \MSP\ and derive the main result. Finally, in Section 5 we show how \MSP\ is related to the no-arbitrage theory in markets with no frictions, with proportional transaction costs, and with convex frictions. We study the models separately, in an increasing level of complexity.

\subsection{Notation used in the paper}\label{notation}
For a given set $A\subset\R^d$, we denote by $\close{A}$, $\inter A$, $\ri A$, $\conv A$, $\cone A$, $\aff A$ and $\lin A$, the closure, the interior, the relative interior, the convex hull, the conical hull, the affine hull and the linear hull  of $A$. Scalar product on $\R^d$ is denoted by $\langle\cdot,\cdot\rangle$. For a cone $A$, we let $A^\ast$ be the (positive) polar of $A$, defined by
$$
  A^\ast\coloneqq\{y\in\R^d\mid\langle y,x\rangle\geq 0\ \ \forall x\in A\}.
$$ 
%which simplifies in an obvious way in case $A$ is a cone.
A map $U$ defined on a state space $\Omega$ and taking values in the power set of $\R^d$ is called a correspondence. We denote that a map is a correspondence by writing $U\colon\Omega\rightrightarrows\R^d$. The domain of $U$ is denoted by $\dom U\coloneqq\{\omega\in\Omega\mid U(\omega)\neq\varnothing\}$. 

Let $\cG$ be a sigma-algebra on $\Omega$ and $U\colon\Omega\rightrightarrows\R^d$ be a correspondence. We say that $U$ is $\cG$ measurable if $\{\omega\,|\,U(\omega)\cap O\not=\varnothing\}\in\cG$ for all open sets $O\subset\R^d$. We denote by $\cL(\cG;U)$ the set of all $\cG$-measurable selections of $U$.

Extended real numbers are denoted by $\eR\coloneqq[-\infty,+\infty]$. Let $f\colon\Omega\times \R^d\rightarrow \eR$ be a function. We denote the domain and the epigraph correspondences by
\begin{align*}
 \dom f(\omega) &\coloneqq \big\{x\in\R^d\mid f(\omega,x)<\infty \big\},\\
 \epi f(\omega) &\coloneqq \big\{(x,\alpha)\in\R^d\times\R\mid f(\omega,x)\leq \alpha \big\}.
\end{align*}

By $\cP$ we denote the set of all finite support probability measures on $\Omega$ and by $\cP(\omega)$ we denote the set of all $Q\in\cP$ such that $\omega\in\supp Q$; by $\supp Q$ we denote the support of the measure $Q$, i.e. the smallest closed set with full $Q$ measure.

%% %% %% %% %% %% %% %% %% %% %% %% %% %% %% %% %% %% %% %% %% %% %% %% %% %% %% %% %% %% 
\section{Setup}\label{Section: setup}

Let $(\Omega,\cB(\Omega))$ be a Polish space endowed with its Borel sigma-algebra. Let $\cI\coloneqq\{0,\ldots, T\}$ with $T\in\N$. Let $(E,\cB(E))$ be a separable metric space endowed with its Borel sigma-algebra and let $(\psi_t)_{t\in\cI}$ be a set of Borel maps
$$
  \psi_t\colon\Omega\longrightarrow E.
$$
We assume that the mapping $\psi_0$ is constant, i.e. there exists a $y\in E$, such that $\psi_0(\omega)=y$ for each $\omega\in\Omega$.

We define a filtration $\widehat{\F}=(\widehat{\cF}_t)_{t\in\cI}$ on $\Omega$: for any $t\in\cI$ we denote by $\widehat\cF_t$ the sigma-algebra generated by maps $\psi_s$ as follows
$$
  \widehat\cF_t = \sigma\big(\{\psi^{-1}_s(A)\mid s\leq t,\ A\in\cB(E)\}\big).
$$
This is the `natural filtration' generated by the process $(\psi_t)$. We will work with the larger filtration $\F\coloneqq(\cF_t)_{t\in\cI}$ given by
$$
  \cF_t =\bigcap_{P\in\cP}\widehat\cF_t\vee\cN^P_t,
  \qquad 
  \cN^P_t\coloneqq\big\{A\subset A'\in\widehat\cF_t \,\big|\, P(A')=0\big\}.
$$
This is a technical condition used to ensure the measurability of the correspondences of Section \ref{sec:proj} below.
Note that the assumption on the map $\psi_0$ implies that $\cF_0=\{\varnothing,\Omega\}$.
% This choice simplifies the exposition, we discuss in the Appendix the possibility of considering a coarser filtration.  

\begin{remark}
  We can define maps $\Psi_t\colon\Omega\rightarrow E^{t+1}$ by
  \begin{align}\label{eq::Psi_t}
    \Psi_t(\omega)=\big(\psi_0(\omega),\ldots,\psi_t(\omega)\big).
  \end{align}
  Since $E$ is a separable metric space, also $\cB(E^{t+1})=\bigotimes_{s=0}^t\cB(E)$. Hence, we may write $\widehat\cF_t = \sigma(\{\Psi^{-1}_s(A)\mid A\in\cB(E^{t+1})\})$.
\end{remark}

The following Lemma identifies the `atoms' of sigma algebras $\cF_t$, defined above, and shows their measurability. Since $\cF_t$ is the $\cP$-completion of $\widehat\cF_t$, it implies that we obtain an $\cF_t$-measurable object by simply specifying its value on every atom. In particular, this will allow us to modify any $\cF_t$ measurable object only on one particular atom and preserve measurability.
\begin{lemma}\label{lem:measurability of Sigma}
The set $\Sigma_t^\omega\subset\Omega$, defined by
\begin{align}\label{eq: level_set}
\Sigma_t^\omega
\coloneqq\big\{
\bar\omega\in\Omega
\,\big|\, 
\psi_s(\bar{\omega}) = \psi_s(\omega),\  \forall s\leq t
\big\},
\end{align}
is Borel measurable for every $\omega\in\Omega$ and $t\in \cI$; in particular $\Sigma^\omega_t\in\cF_t$ for all $\omega$.
\end{lemma}
\begin{proof}
Simply observe that $\Sigma_t^\omega=\Psi_t^{-1}(\tau)$, where $\Psi_t$ is a Borel mapping and $\tau=\Psi_t(\omega)$ is a singleton, thus a closed subset of $E^{t+1}$. The second statement is clear since already $\Sigma^\omega_t\in\widehat\cF_t$ by the definition of $\widehat\cF_t$. 
\end{proof}

\paragraph{Some classical examples in financial mathematics.}
For the convenience of the reader, we anticipate the type of applications that we have in mind. The following are examples of maps $\psi_t$, which collects the most typical models of discrete-time financial markets studied in the literature.

\begin{example}\label{ex:1}
  Assume that we are modeling the frictionless market, given by the stock price process $(S_t)$. Then one can take $E = \R^d$, where $d$ is the number of risky assets, and $\psi_t(\omega)=S_t(\omega)$. This is why we call $\widehat{\F}$ the natural filtration.
\end{example}

\begin{example}
  Similarly to above, one can model the market with a single risky asset by a pair of processes: the bid and ask price process. Those we denote by $(\underline{S}_t,\overline{S}_t)$. Then one takes $E = \R^2$ and $\psi_t(\omega)=(\underline{S}_t(\omega),\overline{S}_t(\omega))$.
\end{example}

%Before going to the next example, let us introduce some terminology; the main reference is~\citep{R}, Chapters~4 and~14. Denote by $\clo(\R^d)$ the set of all closed, non-empty subsets of $\R^d$. One can show that the set $\clo(\R^d)$ admits a metric, with respect to which it is a complete, separable metric space. Furthermore, a set valued map $\kappa\colon\Omega\rightrightarrows\R^d$ is a closed valued $\F$ measurable correspondence if and only if it is measurable as a single valued map into $(\clo(\R^d),\cB(\clo(\R^d)))$; see Theorem 14.4 in~\citep{R}.
% We collect the results about this setup in the Appendix.

\begin{example}
  The basic model of proportional transaction costs, the Kabanov's model of currency markets, is given by a family of closed, convex solvency cones $(K_t)_{t\in\cI}$. The set $K_t(\omega)$ represents the set of solvent positions at time $t$ if $\omega$ occurs. If the correspondence $K_t\colon\Omega\rightrightarrows\R^d$ is Borel measurable, then it is also a Borel measurable map with values in the metric space $\clo(\R^d)$, of closed subsets of $\R^d$, with its Borel sigma algebra; see~\citep{R}, chapters~4 and~14. Now we may define $\psi_t=K_t$ and the filtration $\F$ as above.
\end{example}
\begin{remark}
Throughout the paper all the characterization results are given using finite support probability measures.
In particular, this is the case for the versions of the FTAP of Theorem \ref{thm:FTAPfrictionless}, Theorem \ref{thm: FTAP Kabanov} and Theorem \ref{thm: pennanen}.
To get an intuition why this is enough one can think in terms of a frictionless market as in Example \ref{ex:1}.
Essentially, the existence of an arbitrage opportunity reduces to the question of whether $0$ can be separated from the increments of $(S_t)$ by an hyperplane.
In a pointwise framework this is a question regarding only the geometry of the price process rather than the support of the desired martingale measures. 
   
\end{remark}

\section{Projections of measurable correspondences}\label{sec:proj}
Analyzing the martingale selection theorem will follow dynamic programming ideas. Therefore, the first step is to have an object that generalizes the `conditional support' of a correspondence to this, pointwise, setting.
Let $X$ be an $\cF_{t+1}$-measurable random variable and $t\in\cI\backslash\{T\}$.
In the classical case a reference measure $P$ is given and the conditional support of $X$, given $\cF_t$, is the smallest closed valued, $\cF_t$-measurable, correspondence $A$ such that $X$ is a selection of $A$ $P$-a.s. 
The existence of such a conditional support and that of regular versions of the conditional distribution are instrumental for the approach of~\citep{Rok06,Rok}.
In this pointwise setting one generalizes the conditional support as follows: 
\begin{align}\label{eq:essinf}
  X_t(\omega)
  \coloneqq
  \big\{\E_Q[X|\cF_t](\omega)\,\big|\, Q\in\cP(\omega)\big\}.
\end{align}
The correspondence $X_t\colon\Omega\rightrightarrows\R^d$ is well defined by definition of $\cP(\omega)$ (see Section \ref{notation}); indeed, $\{\omega\}$ is an atom for every $Q\in\cP(\omega)$. Moreover, by the choice of $\cF_t$, it is also measurable. Clearly, $X(\omega)\in X_t(\omega)$ for every $\omega\in\Omega$.

\begin{lemma}\label{lem:support correspondence}
  Let $t\in\cI$ and let $U\colon\Omega\rightrightarrows\R^d$ be an $\cF_{t+1}$-measurable correspondence. Then the correspondence
\begin{align*}
  U^{\sharp}(\omega)
  \coloneqq
  \conv\big(\big\{x\in\R^d\,\big|\,x\in U(\bar\omega) \, \text{ for some }\, \bar{\omega}\in\Sigma_{t}^\omega\big\}\big)
\end{align*}
is $\cF_t$-measurable and convex valued.
\end{lemma}
\begin{proof}
Obvious, since $U^{\sharp}(\omega)$ depends only on $\Sigma_t^\omega$.
\end{proof}

\begin{remark}
  We are assuming neither closed nor convex values of the correspondence $U$ in the lemma above. Since we want, ultimately, that our correspondences have convex values, we define $U^\sharp$ immediately as such. Let us rewrite it as follows
\begin{align}\label{eq:sharp}
  U^\sharp(\omega)=\conv\bigcup_{\bar\omega\in\Sigma_t^\omega}U(\bar\omega).
\end{align}
This makes it clear that the convex hull operation is necessary. Note, however, that $U^\sharp$ need not to be closed. Moreover, one can write it as
\begin{align}\label{rmk: char U sharp}
   U^\sharp(\omega) = \big\{y%\in\R^d 
   \,\big|\,
   \exists\bar\omega\in\Sigma_t^\omega, \, \xi\in \mathcal{L}(\cF_{t+1};U),\ Q\in \mathcal{P}(\bar\omega)\colon  \E_Q[\xi|\cF_t](\omega)=y \big\};
\end{align}
this is a direct consequence of the Carath\'eodory's theorem on convex hulls.
\end{remark}

The representation \eqref{rmk: char U sharp} highlights the connection between $U^\sharp$ and the classical notion of conditional expectation of a correspondence which is defined through the conditional expectation of its selections; see~\citep{Mol}.
For any measure $P\in\cP$, the conditional expectation of any selection of $U$ will be a selection of $U^\sharp$. 
However, in general, $U^\sharp$ will be a larger set as it includes the $P$-conditional expectation of $U$ with respect to \emph{any} $P\in\cP$ (on the support of the measure).

The correspondence $U^\sharp$ is too big for our needs; this will become evident in Section~\ref{Section: MSP}. We, therefore, define a smaller set-valued map
\begin{align}\label{eq: Uflat}
  U^\flat(\omega) 
  \coloneqq
  \bigcap_{\bar\omega\in\Sigma^\omega_t}\bigcup_{\lambda\in(0,1)}\big[\lambda U(\bar{\omega})+(1-\lambda) U^\sharp(\omega)\big].
\end{align}
\begin{lemma}\label{lem:regular support correspondence}
  Let $t\in\cI\backslash\{T\}$ and let $U\colon\Omega\rightrightarrows\R^d$ be a convex valued $\cF_{t+1}$-measurable correspondence. Then $U^\flat(\omega)$ is convex valued and $\cF_t$-measurable.
\end{lemma}

\begin{proof}
  For $\bar\omega\in\Sigma^\omega_t$ we define the set 
  $$
    \widehat U(\bar\omega)\coloneqq\bigcup_{\lambda\in(0,1)}\big[\lambda U(\bar{\omega})+(1-\lambda) U^\sharp(\omega)\big].
  $$
  To show that the correspondence $U^\flat$ is convex valued, it is enough to observe that the set $\widehat U(\bar\omega)$ is convex for each $\bar\omega$. Since $\ri U^\sharp(\omega)\subset \widehat U(\bar\omega)\subset U^\sharp(\omega)$ for every $\bar\omega\in\Sigma^\omega_t$, it follows directly that $\ri U^\sharp(\omega)\subset U^\flat(\omega)\subset U^\sharp(\omega)$. Hence, $U^\flat$ is measurable by Lemma~18.3 in \citep{Aliprantis}. 
\end{proof}

\begin{remark}
We will use the $(\cdot)^\sharp$ and $(\cdot)^\flat$ operators, on elements of adapted sequences of correspondences $W$. Thus, measurability of the resulting correspondence will be clear: $W^\sharp_{t+1}$ is $\cF_t$ measurable, since $W_{t+1}$ is $\cF_{t+1}$ measurable; the same for $W^\flat_{t+1}$.
\end{remark}

We showed in the proof of the above lemma that 
\begin{align}\label{eq:inclusions}
  \ri U^\sharp(\omega)\subseteq U^\flat(\omega)\subseteq U^\sharp(\omega),
\end{align}
i.e. $U^\flat$ is obtained from $U^\sharp$ by, possibly, omission of some points from the relative boundary of $U^\sharp$. The significance of this construction is contained in the following statement. 

\begin{lemma}\label{lem:key lemma}
  Let $t\in\cI\backslash\{T\}$ and let $U\colon\Omega\rightrightarrows\R^d$ be a convex valued $\cF_{t+1}$-measurable correspondence. Then
   $$
    U^\flat(\omega) = \big\{y\in\R^d \,\big|\, \forall\bar\omega\in\Sigma_t^\omega\ \exists\,\xi\in \mathcal{L}(\cF_{t+1};U),\ Q\in \mathcal{P}(\bar\omega)\colon  \E_Q[\xi|\cF_t](\omega)=y \big\}.
  $$
\end{lemma}

In words, $y\in U^\flat(\omega)$ if and only if for \emph{every} $\bar\omega\in\Sigma_t^\omega$ there exists a selection $\xi$ of $U$, and a measure $Q\in\cP(\bar\omega)$, such that $\E_Q[\xi|\cF_t](\omega)=y$. The crucial difference with $U^\sharp$ can be seen by comparing a similar characterization, given in equation~\eqref{rmk: char U sharp}. %We observe that in ~\eqref{rmk: char U sharp} 
There, the union of the supports of the measures $Q$ does not necessarily contain every $\bar \omega\in\Sigma_t^\omega$, but only at least one.

\begin{proof}
Note that if $\Sigma_t^\omega=\{\omega\}$ there is nothing to prove. Otherwise, let us start by showing the inclusion $\subseteq$. Fix $\bar{\omega}\in\Sigma_t^\omega$. Since $y\in U^\flat(\omega)$, there exists a $\lambda\in(0,1)$ such that $y\in[\lambda U(\bar{\omega})+(1-\lambda) U^\sharp(\omega)]$. The rest of the claim follows from Carath\'eodory's theorem. Indeed, we may write $y = \lambda x + (1-\lambda)z$ for $x\in U(\bar\omega)$ and $z\in U^\sharp(\omega)$. Since $U^\sharp(\omega)$ is a convex hull of a union, we may write $z = \lambda_1z_1+\cdots+\lambda_nz_n$ for some $z_i\in U(\omega_i)$, $\omega_i\in\Sigma^\omega_t$, and $\lambda_i\in(0,1)$, such that $\lambda_1+\cdots+\lambda_n=1$. We may assume that $x$ and $z_i$ in this decomposition are chosen such that $n$ is minimal. This implies, in particular, that the sets $\Sigma^{\bar\omega}_{t+1}$ and $\Sigma^{\omega_i}_{t+1}$ are pairwise disjoint. Then 
  $$
    Q=\lambda\delta_{\bar \omega} +(1-\lambda)\sum_{i=1}^n\lambda_i\delta_{\omega_i}
    \qquad\text{and}\qquad 
    \xi(\omega) = x\one_{\Sigma^{\bar\omega}_{t+1}} + \sum_{i=1}^nz_i\one_{\Sigma^{\omega_i}_{t+1}},
  $$
  where we denoted by $\delta_\omega$ the Dirac measure with mass concentrated in $\omega$. Note that $\xi$ can be extended in an arbitrary way to a selection of $U$ on the complement of $\Sigma_{t+1}^{\bar\omega}\bigcup_i\Sigma_{t+1}^{\omega_i}$; see Lemma~\ref{lem:selection ri}, Lemma~\ref{lem:measurability of Sigma} and the comment above it.

  As for the converse, we need to show that for any element $y\notin U^\flat(\omega)$ there exists an $\bar\omega\in\Sigma^\omega_t$, such that as soon as the pair $(\xi,Q)$ satisfies $Q[\Sigma^{\bar\omega}_{t+1}]>0$, we have $\E_Q[\xi|\cF_t](\omega)\neq y$. If $y\notin U^\sharp(\omega)$, it is clear from~\eqref{rmk: char U sharp}, that any $\bar\omega\in\Sigma^\omega_t$ will satisfy the desired properties. Assume, therefore, that $y\in U^\sharp(\omega)\backslash U^\flat(\omega)$. Since the set $U^\flat(\omega)$ is convex and $\close{U^\flat}=\close{U^\sharp}$, $y$ is a point in the relative boundary of $U^\sharp(\omega)$. Denote by $F$ the extremal face of $U^\sharp(\omega)$ containing $y$, i.e.
  $$
    F = \big\{z\in U^\sharp(\omega)\,\big|\,\exists x\in U^\sharp(\omega),\lambda\in(0,1):\ y = \lambda z + (1-\lambda)x\big\}.
  $$
  If $U(\bar{\omega})\cap F\not=\varnothing$ for each $\bar\omega\in\Sigma^\omega_t$, then, by the definition of $U^\flat(\omega)$, we have $y\in U^\flat(\omega)$, which contradicts our assumption. Therefore, there exists an $\bar{\omega}\in\Sigma^\omega_t$ such that $U(\bar{\omega})\cap F=\varnothing$. It is easy to see that for any selection $\xi$ of $U$ and finite support measure $Q$, such that $\bar\omega\in\supp Q$, we have that $\E_Q[\xi|\cF_t](\omega)\not\in F$.
\end{proof}

The representation of Lemma \ref{lem:key lemma} shows a fundamental robustness property of the correspondence $U^{\flat}$.
Indeed, an element $y\in U^{\flat}$ is not only the conditional expectation with respect to a certain $P\in\cP$ but, using the convexity property of Lemma \ref{lem::convexity} below, it is also possible to enlarge the support of $P$ with an arbitrary $\bar{\omega}\in\Sigma_t^{\omega}$ and still find a selection of $U$ with the same conditional expectation $y$.
In relation to the martingale selection problem, this will be instrumental for preventing that for every $(\xi,Q)$ solution to the $\MSP$, the support of $Q$ must be limited to a certain subset of $\Omega$.

Although, in general, $\ri(U^\sharp)\not=U^\flat$, there are easy conditions ensuring equality; cf. equation~\eqref{eq:inclusions}. This is the content of the following two lemmas.

\begin{lemma}\label{cor: sharp vs flat}
 Let $U\colon\Omega\rightrightarrows\R^d$ be an $\cF_{t+1}$ measurable correspondence with open values. Then the correspondence $U^\sharp$ has open values. In particular $U^\flat=U^\sharp$.
\end{lemma}
\begin{proof} 
An easy consequence of the Carath\'eodory's theorem is that the convex hull of an open set is open. Thus, from equation~\eqref{eq:sharp}, $U^\sharp$ is open. The final statement follows from equation~\eqref{eq:inclusions}.
\end{proof}

\begin{lemma}\label{lem: ri flat}
 Let $U\colon\Omega\rightrightarrows\R^d$ be an $\cF_{t+1}$-measurable correspondence with convex and relatively open values. Then $U^\flat=\ri(U^\sharp)$; in particular, $U^\flat$ has relatively open values.
\end{lemma}
\begin{proof}
Fix $\omega\in\Omega$. Argue by contradiction and choose $x\in U^\flat(\omega)\backslash \ri U^\sharp(\omega)$. There exists a halfspace $\cH\coloneqq\{z\in\R^d\,|\, \langle h,z\rangle\geq\alpha\}$ for some $h\in\R^d$ and $\alpha\in\R$, such that $U^\sharp(\omega)\subset\cH$, $\langle h,x\rangle=\alpha$ and such that $U^\sharp(\omega)\cap\inter\cH\not=\varnothing$. By the choice of $x$,
$$
  x\in \bigcup_{\lambda\in(0,1)}[\lambda U(\bar\omega) + (1-\lambda)U^\sharp(\omega)]\qquad \forall \bar\omega\in\Sigma^\omega_t.
$$
This implies that $U(\bar\omega)\cap \{z\,|\,\langle h,z\rangle=\alpha\}\not=\varnothing$ for every $\bar\omega\in\Sigma^\omega_t$. Since $U$ is relatively open, it needs to be $U(\bar\omega)\subset \{z\,|\,\langle h,z\rangle=\alpha\}$; see Theorem~18.1 in~\cite{R70}. Since $\bar\omega\in\Sigma^\omega_t$ is arbitrary, the same inclusion holds for $U^\sharp(\omega)$. This contradicts the assumption that $U^\sharp(\omega)\cap\inter\cH\not=\varnothing$. 
\end{proof}

\section{The martingale selection theorem and the main result}\label{Section: MSP}

In this section we will define the martingale selection problem (\MSP). The problem was initially studied by Rokhlin in a series of papers, see e.g.~\citep{Rok06,Rok}. Start with the Polish  space $\Omega$ with the Borel $\sigma$-algebra $\cB(\Omega)$ and the filtration $\F=(\cF_t)_{t\in\cI}$ defined as in the previous section.

\begin{definition}
Let $V$ and $C$ be two $\F$-adapted sequences of correspondences $\Omega\rightrightarrows\R^d$ such that 
\begin{align*}
  V&=(V_t)_{t\in\cI}\ \mbox{ has relatively open, convex values;}\\
  C&=(C_t)_{t\in\cI}\ \mbox{ has closed, convex values.}
\end{align*}
Such a pair $(V,C)$ we call a martingale selection problem (\MSP).
\end{definition}

\begin{definition}\label{def:solvability}
We say that the the martingale selection problem $(V,C)$ is \emph{solvable} if for every $\bar\omega\in\Omega$ there exists an $\F$-adapted process $\xi=(\xi_t)_{t\in\cI}$ and a probability measure $Q\in\cP(\bar\omega)$ such that $\xi_t\in\cL(\cF_t;V_t)$ and
\begin{align}\label{eq: martingale}
\E_Q\big[\xi_{t+1}-\xi_{t}\big|\cF_t\big]\in C_t\quad Q\text{-a.s.} \ \mbox{for all }t\in\cI\backslash\{T\}.
\end{align}
We call the pair $(\xi,Q)$ with $Q\in\cP(\bar\omega)$ a \emph{local solution of $(V,C)$ at $\bar\omega$.}
\end{definition}

\begin{remark}
  Let $(\xi,Q)$ be any solution to the \MSP: $(V,C=(\{0\}))$. Condition~\eqref{eq: martingale} states that $\xi$ is a $Q$ martingale; recall that $Q$ is a finite support measure, hence $\xi$ is integrable. This is where the name `martingale selection problem' comes from. 
  The terminology `local' in the above definition aims at emphasizing the fact that a given $\bar{\omega}$ is in the support of $Q$.
  When we do not need that a specific $\bar{\omega}$ belongs to the support, we will simply call $(\xi,Q)$ a solution to the $\MSP$.
\end{remark}

Before proceeding further with the analysis of the \MSP, let us first state a basic observation about the set of (local) solutions. It will prove instrumental in the proof of the main theorem. The statement is, essentially, that the set of solutions enjoys a form of convexity property.
\begin{lemma}\label{lem::convexity}
  Let $(\xi^k,Q^k)$ be solutions to the \MSP\ $(V,C)$ for $k=1,\ldots,n$. Then, for any convex combination $Q$ of measures $(Q^k)$ there exists a process $\xi$, such that $(\xi,Q)$ is a solution to the \MSP.
\end{lemma}
\begin{proof}
  Write $Q=\sum_{k=1}^n\lambda_k Q^k$, where $\lambda_k\in(0,1)$ for all $k$ and $\lambda_1+\cdots+\lambda_n=1$. 
  Let us show that the process
  \begin{align}\label{eq::convexity}
  \xi_t(\omega)=\sum_{k=1}^n \frac{\lambda_kQ^k[\Sigma^\omega_t]}{Q[\Sigma^\omega_t]}\xi^k_t(\omega)
  \end{align}
  together with the measure $Q$ solves the \MSP. Note that we can extend the process $\xi_t$ outside of the support of the measure $Q$ by setting $\frac00=1$ in~\eqref{eq::convexity}. First, it is clear that the process $\xi$ is adapted and that it is a selection of $V$; Indeed, a cursory inspection of the definition of $\xi$ will reveal the sum to be a convex combination. Moreover, measurability follows from $\Sigma^\omega_t\in\cF_t$. Hence, the only thing to prove is that it satisfies~\eqref{eq: martingale}. Calculate
  \begin{align*}
    \E_Q[\xi_{t+1}-\xi_{t}|\cF_t](\omega)
    &= 
    \sum_{\Sigma^{\bar\omega}_{t+1}\subset\Sigma^{\omega}_t}\frac{Q[\Sigma^{\bar\omega}_{t+1}]}{Q[\Sigma^{\omega}_t]}\xi_{t+1}(\bar\omega)
    -\xi_t(\omega)
    \\[1ex]&=
    \sum_{\Sigma^{\bar\omega}_{t+1}\subset\Sigma^{\omega}_t}\ \sum_{k=1}^n \frac{\lambda_kQ^k[\Sigma^{\bar\omega}_{t+1}]}{Q[\Sigma^\omega_t]}\xi_{t+1}^k(\bar\omega)
    -\xi_t^k(\omega)
    \\[1ex]&= 
    \sum_{k=1}^n \frac{\lambda_kQ^k[\Sigma^\omega_t]}{Q[\Sigma^\omega_t]}\left[\sum_{\Sigma^{\bar\omega}_{t+1}\subset\Sigma^{\omega}_t}\frac{Q^k[\Sigma^{\bar\omega}_{t+1}]}{Q^k[\Sigma^{\omega}_t]}\xi_{t+1}^k(\bar\omega)
    -\xi_t^k(\omega)\right]
    \\[1ex]&=
    \sum_{k=1}^n \frac{\lambda_kQ^k[\Sigma^\omega_t]}{Q[\Sigma^\omega_t]}\E_{Q^k}\big[\xi^k_{t+1}-\xi^k_t\big|\cF_t\big](\omega).
  \end{align*}
  Notice that the last sum is a convex combination of elements of $C_t(\omega)$. Hence, we conclude by convexity of $C$.
\end{proof}

\begin{remark}
  Observe the definition of $\xi$ in the proof of the Lemma above. One sees that this is given by
  $$
    \xi_t(\omega)=\sum_{k=1}^n \lambda_k\frac{\d Q^k}{\d Q}\bigg|_{\cF_t}\xi^k_t(\omega),
  $$
  where by $\frac{\d Q^k}{\d Q}$ we denote the Radon-Nikodym derivative.
\end{remark}

\begin{remark} Observe that, in general, there is no unique local solution to the $\MSP$ at a given $\bar \omega$.
Indeed, unless the space $\Omega$ is a finite number of events, there is no measure which assign positive mass to every $\omega\in\Omega$. Therefore, given $(\xi^1,Q^1)$ a local solution at $\bar{\omega}$ and $(\xi^2,Q^2)$ a local solution at a certain $\hat{\omega}\notin \supp(Q^1)$, Lemma \ref{lem:key lemma} yields a new local solution $(\xi,Q)$ at $\bar{\omega}$.
\end{remark}

\paragraph{The main theorem.}

Consider a martingale selection problem $(V,C)$. Define the following (adapted) sequence $W=(W_t)_{t\in\cI}$ of measurable correspondences: Set $W_T\coloneqq V_T$ and
\begin{align}\label{eq: iteration}
  W_t\coloneqq V_t\cap \big(W^\flat_{t+1}-C_t\big)
  \qquad \text{ for }t=T-1,\ldots, 0.
\end{align}
The following is the main result of the paper.

\begin{theorem}\label{thm: main1}
	The martingale selection problem $(V,C)$ is solvable if and only if $W_t(\omega)\neq\varnothing$ for all $t\in\cI$ and $\omega\in\Omega$.
\end{theorem}

We prove the result in several steps. Sufficiency is a consequence of the following lemma. 
 
\begin{lemma}\label{lem:extension}
  Let $\bar\omega\in\Omega$ and $t\in\cI$. Assume that $W_s(\omega)\not=\varnothing$ for every $\omega\in\Omega$ and every $s=t,\ldots,T$. Then, for every $\xi\in W_t(\bar\omega)$ there exists a $Q\in\cP(\bar\omega)$ and a process $(\xi_s)_{s=t,\ldots,T}$ with $\xi_s\in\mathcal{L}(\cF_s;V_s)$ for every $s=t,\ldots,T$, such that $\xi_t(\bar\omega)=\xi$ and
  $$
    \E_Q\big[\xi_{s+1}-\xi_s\big|\cF_s\big]\in C_s\quad Q\text{\textnormal{-a.s.}}\qquad s=t,\ldots,T-1.
  $$
\end{lemma}

\begin{proof}
  If $t=T$ there is nothing to prove, one simply chooses $(\xi,\delta_{\bar\omega})$.
  
  Let us assume that the result is true for $t+1$ and let us show it for $t$. We can write each $\xi\in W_t(\bar\omega)$ as
  $$
  \xi = w - c \qquad \mbox{with } w\in W_{t+1}^\flat(\bar\omega),\ \ c\in C_t(\bar\omega).
  $$
  There exists a measure $\bar Q\in\cP(\bar\omega)$ and a random vector $\psi\in\cL(\cF_{t+1};W_{t+1})$ such that $w=\E_{\bar Q}[\psi|\cF_t](\bar\omega)$; see Lemma~\ref{lem:key lemma}. We may assume that $\supp Q\subset\Sigma_t^{\bar\omega}$. Denote by $\omega^i$, $i=1,\ldots,p$, the atoms of~$\bar Q$. By convexity of the correspondences, we can assume that $\Sigma^{\omega^i}_{t+1}$, $i=1,\ldots,p$, are pairwise disjoint. By the induction hypothesis, there exists a solution $(\xi^i,Q^i)$ to the \MSP\ such that $Q^i\in\cP(\omega^i)$ and $\xi^i_{t+1}(\omega^i) =\psi(\omega^i)$. We may, furthermore, assume that $\supp Q^i\subset\Sigma_{t+1}^{\omega^i}$. Define the process~$(\xi_s)$ by
  $$
  \xi_{t}(\bar\omega) = \xi,
  \qquad
  \xi_s = \sum_{i=1}^p\xi^i_s\one_{\Sigma^{\omega^i}_{t+1}}\quad\text{for }s=t+1,\ldots,T,
  $$
  where we extend it to any selection of $W_{t+1}$ on the complement of the considered sets. The measure $Q$ is defined by
  $$
    Q[A]=\sum_{i=1}^p\bar Q\big[\Sigma^{\omega^i}_{t+1}\big]Q^i[A]\qquad A\in\cF_T.
  $$
  Note that $\bar{\omega}$ is an atom of $\bar{Q}$ by assumption. This finishes the proof.
\end{proof}

\paragraph{The case when $V_t$ are open.}
To show necessity in the main theorem, we first study a simpler case. The following observation is the main reason why considering open valued correspondences makes the problem significantly easier.
\begin{lemma}\label{lem:openness}
  Let $V_t(\omega)$ be open for every $\omega$, $t$. Then $W_t(\omega)$ are open.
\end{lemma}
\begin{proof}
 If $W_t(\omega)=\varnothing$, then it is open by definition. Let us, thus, assume the contrary. Observe that, by definition, $W_T(\omega)$ is open for every $\omega$. The statement follows by induction: assume that the same is true for $s=t+1$ and let us show it for $s=t$. Since $W_{t+1}(\omega)$ is open for every $\omega$, so is $W_{t+1}^\sharp(\omega)=W_{t+1}^\flat(\omega)$ by Lemma~\ref{cor: sharp vs flat}. The set $W_t(\omega)$ is therefore an intersection of open sets, hence open.
\end{proof}

  We now want to show that, if the \MSP\ is solvable, $W_t(\omega)\not=\varnothing$ for all $t\in\cI$ and $\omega\in\Omega$. The idea is simple: we will show that any local solution $(\xi,Q)$ to the \MSP\ with $Q\in\mathcal{P}(\omega)$ satisfies $\xi_t(\omega)\in W_t(\omega)$.

\begin{lemma}\label{lem:inclusion}
  Let $(\xi,Q)$ be a local solution to the \MSP\ $(V,C)$ at $\omega\in\Omega$. Then $\xi_s(\omega)\in W_s(\omega)$ for all $s\in\cI$. Consequently, $W_s(\omega)\not=\varnothing$ for all $s\in\cI$.
\end{lemma}
\begin{proof}
  First note that, in light of Lemma~\ref{lem:openness}, the sequence $W_t$ could have been defined as follows: $W_T\coloneqq V_T$ and
  $$
    W_t\coloneqq V_t\cap \big(W^\sharp_{t+1}-C_t\big) \qquad \text{ for }t=T-1,\ldots, 0.
  $$
  Let $(\xi,Q)$ be the local solution to the \MSP\ at $\omega$. By definition, $\xi_s(\omega)\in V_s(\omega)$ and also 
  $$
    \E_Q\big[\xi_{s+1}-\xi_s\big|\cF_s\big](\omega)\in C_s(\omega) \qquad \forall s=0,\ldots,T-1.
  $$
  But this last expression can be read out as $\xi_s(\omega)\in (\E_Q[\xi_{s+1}|\cF_s](\omega) - C_s(\omega))$.
  
  Now we come to the induction argument: clearly, $\xi_T(\bar\omega)\in W_T(\bar\omega)$ for each $\bar{\omega}\in\Omega$ and in particular for every $\bar{\omega}\in\supp Q$. Assume that $\xi_{t+1}(\bar\omega)\in W_{t+1}(\bar\omega)$ for each $\bar{\omega}\in\supp Q$. Then, noticing that $\E_Q[\xi_{s+1}|\cF_s](\bar\omega)\in W_{t+1}^\sharp(\bar{\omega})$, we have
  $$
   \xi_t(\bar\omega)
   \in
   \big(\E_Q[\xi_{t+1}|\cF_t](\bar\omega) - C_t(\bar\omega)\big)
   \subset
   W^\sharp_{t+1}(\bar\omega) - C_t(\bar\omega).
  $$
  This proves the claim.
\end{proof}

We come to the proof of the main theorem.
\begin{proof}[Proof of the main theorem for $V$ open.]
  One implication follows directly from Lemma~\ref{lem:extension}. The other follows from Lemma~\ref{lem:inclusion}.
\end{proof}

\paragraph{The general case.} 
When the correspondence $V_t$ is open valued, we showed that $W^\sharp=W^\flat$ holds and $\E_Q[\xi_{s+1}|\cF_s]\in W^\sharp_s$ $Q$-a.s. for every solution $(\xi,Q)$ to the \MSP. In general, $W^\sharp$ does not necessarily coincide with $W^\flat$, thus we may fail to have $\E_Q[\xi_{s+1}|\cF_s]\in W^\flat_s$ $Q$-a.s. for some solution $(\xi,Q)$ to the \MSP.

\begin{example}\label{ex::flat vs sharp}
  Consider $T=1$ and $\Omega=[0,1]$ with $d=1$. Define $C_0=\{0\}$ and $V_1(\omega)=\{\omega\}$ and $V_0=\{0\}$. Then one easily sees that $W_1^\sharp=[0,1]$ (which is not an open set). Nevertheless, one can compute that $W_1^\flat=(0,1)\neq W_1^\sharp$ and, thus, $W_0=\varnothing$. Notice that \MSP\ is not solvable in this example; Indeed, if $(\xi,Q)$ is a local solution of the \MSP\ at $\omega$, we necessarily have $\omega=0$, in other words, $Q=\delta_0$ is the only martingale measure.
\end{example}

\begin{example}
  \cite{Rok} proposes a different iteration for the sequence of measurable correspondences $W_t$ which we denote by $w$. Define $w_T\coloneqq V_T$ and
  \begin{align*}
  w_t&\coloneqq V_t\cap \big(\ri w^\sharp_{t+1}-C_t\big) \qquad \text{ for }t=T-1,\ldots, 0.
  \end{align*}
  He claims, albeit without proof, that $w_t(\omega)\not=\varnothing$ for all $t\in\cI$ and $\omega\in\Omega$ is equivalent to the \MSP\ being solvable.
  Observe, however, the following example: let $\cF_i$ be trivial for each $t=0,1,2$. Define $V_0=V_2=(-1,1)\times\{0\}$ and $V_1=(-1,1)^2$. Define also $C_0=\{0\}$ and $C_1=\{(x,y)|y\geq0\}$. One easily gets $w_1=(-1,1)\times(-1,0]$ and $w_0=\varnothing$. However, \MSP\ is clearly solvable, e.g. take any constant process $\xi$. This implies that the sets $w_t$ are too small, and this motivated the definition of the $(\cdot)^\flat$ operation.
\end{example}

\begin{remark}\label{rmk: ri Wflat}
As the previous example indicates, $W^\flat_t$ are not necessarily relatively open, even if $V_t$ are. In the previous example we get $W_1 = w_1$, which is not relatively open. We also get $W_0=V_0$.
 \end{remark}

In the following we show that the iteration based on the $(\cdot)^\flat$ operation yields the result.
We start by providing a new characterization of the sets $W_s$ which is a simple corollary to Lemma~\ref{lem:key lemma}.

\begin{lemma}\label{lem::characterization}
Let $t\in\cI\backslash\{T\}$ and assume that $W_{t+1}(\omega)\not=\varnothing$ for every $\omega\in\Omega$. Then
$$
W_{t+1}^\flat(\omega)-C_t(\omega) = 
\Bigg\{
y\in \R^d\,
\Bigg|\, 
\begin{array}{l}
\forall\bar\omega\in\Sigma_t^\omega\ \exists\,\xi\in \cL(\cF_{t+1};W_{t+1}),\ Q\in \cP(\bar\omega):
\\[1ex] 
\E_Q[\xi|\cF_t](\omega)-y\in C_t(\omega)
\end{array} 
\Bigg\}.
$$
\end{lemma}

\begin{proof}
  To show $\subset$ fix an $\bar\omega\in\Sigma_t^\omega$. Let $ y\in W_{t+1}^\flat(\omega)-C_t(\omega)$ be arbitrary and write it as $y=w-c$, where $w\in W_{t+1}^\flat(\omega)$ and $c\in C_t(\omega)$. The rest follows from Lemma~\ref{lem:key lemma}.
	
  As for the converse, let $y$ be an element of the set on the right hand side. The set $C_t(\omega)\cap(W_{t+1}^\sharp(\omega)-y)$ is convex as it is an intersection of convex sets. To see that it is also nonempty, let $\bar\omega\in\Sigma^\omega_t$ be arbitrary and let $\xi$ and $Q$ be from the definition of the right hand side. Then $\E_Q[\xi|\cF_t]-y\in C_t(\omega)\cap(W_{t+1}^\sharp(\omega)-y)$. We claim that $y+c\in W_{t+1}^\flat(\omega)$ for every $c\in\ri(C_t(\omega)\cap(W_{t+1}^\sharp(\omega)-y))$, from which the result follows.
  
  Proceed as in the proof of Lemma~\ref{lem::convexity}. Choose first an arbitrary $\omega^1\in\Sigma^\omega_t$ and the corresponding $\xi^1\in\cL(\cF_{t+1};W_{t+1})$ and $Q^1\in\cP(\omega^1)$, $Q^1[\Sigma^\omega_t]=1$, such that $c^1 = \E_{Q^1}[\xi^1|\cF_t](\omega)-y$. Then, by the choice of $c$ in the relative interior, there exists a $\xi^2\in\cL(\cF_{t+1};W_{t+1})$ with $Q^2$, a finite support measure with $Q^2[\Sigma^\omega_t]=1$, such that $c^2 = \E_{Q^2}[\xi^2|\cF_t](\omega)-y$ and $c = \lambda c^1+(1-\lambda)c^2$. Then, by choosing the pair
  $$
    Q=\lambda Q^1+(1-\lambda)Q^2,
    \qquad
    \xi(\omega)
    =
    \lambda\frac{Q^1[\Sigma^\omega_{t+1}]}{Q[\Sigma^\omega_{t+1}]}\xi^1(\omega)
    +
    (1-\lambda)\frac{Q^2[\Sigma^\omega_{t+1}]}{Q[\Sigma^\omega_{t+1}]}\xi^2(\omega)
  $$
  we have $y+c\in \lambda U(\omega^1)+(1-\lambda)U^\sharp(\omega)$. Note that we use $\frac00=1$ in the last equation. From $\omega^1$ being arbitrary and from \eqref{eq: Uflat} the result follows.
\end{proof}

\medskip

To prove the main theorem it remains to show that if the \MSP\ $(V,C)$ is solvable, then also the set $W_t(\omega)$ is nonempty for every $t\in\cI$ and every $\omega\in\Omega$. To this aim, we need to connect the solution of the \MSP\ to sets $W_t(\omega)$. We define for every $t\in\cI$ and $\omega\in\Omega$ the following set
$$
  \cT_t(\omega) \coloneqq \big\{\xi_t(\omega)\,\big|\,(\xi,Q)\mbox{ local solution to \MSP\ at } \omega\big\}\subset V_t(\omega).
$$
 It is nonempty by the assumption that \MSP\ is solvable. To see that $\cT_t(\omega)$ is convex, choose two local solutions $(\xi^1,Q^1)$ and $(\xi^2,Q^2)$ to the \MSP\ $(V,C)$ at $\omega$. For any $\mu\in(0,1)$ set $\lambda = \frac{\mu Q^2[\Sigma_t^\omega]}{\mu Q^2[\Sigma_t^\omega]+(1-\mu) Q^1[\Sigma_t^\omega]}$. By Lemma~\ref{lem::convexity}, there is a process $\xi$ such that the pair $(\xi,Q)$, with $Q=\lambda Q^1+(1-\lambda)Q^2$, is a solution to the \MSP; the process $\xi$ is given in equation~\eqref{eq::convexity}. The evaluation yields
$$
  \xi_t(\omega) 
  = 
  \lambda\frac{Q^1[\Sigma_t^\omega]}{Q[\Sigma_t^\omega]}\xi^1_t(\omega)
  +(1-\lambda)\frac{Q^2[\Sigma_t^\omega]}{Q[\Sigma_t^\omega]}\xi^2_t(\omega)
  = \mu \xi^1_t(\omega) + (1-\mu)\xi^2_t(\omega).
$$

Example~\ref{ex::flat vs sharp} shows that $\cT_t(\omega)\not\subset W_t(\omega)$ in general. To establish the main theorem it is enough to prove $\ri\cT_t(\omega)\subset W_t(\omega)$ for every $\omega$, $t$. We are going to prove this inclusion by showing that
\begin{align}\label{eq::to prove inclusion}
    \ri\cT_t(\omega)\subset (W_{t+1})^\flat(\omega)-C_t(\omega).
\end{align}
Since $\ri \cT_T\subset V_T=W_T$, by showing \eqref{eq::to prove inclusion} we also have $\ri\cT_t\subset W_{t}$ which is therefore non-empty.
\begin{remark}\label{lem::remark on measurability of cT}
%Note that we are not claiming that the set valued maps $\omega\mapsto\cT_t(\omega)$ are measurable.  
We will prove~\eqref{eq::to prove inclusion}  by showing that $\ri\cT_t(\omega)\subset (\ri\cT_{t+1})^\flat(\omega)-C_t(\omega)$ holds, where the objects are defined purely algebraically.
\end{remark}

\begin{proof}[Proof of Theorem \ref{thm: main1}.]
  If $W_t(\omega)\not=\varnothing$ for all $\omega$ and $t$, Lemma~\ref{lem:extension} implies that the \MSP\ is solvable.

We now prove~\eqref{eq::to prove inclusion}. Proceed with a sequence of easy observations.
 
  \textsc{Step~1:} Let $y\in \ri\cT_t(\omega)$. We claim that for every $\bar\omega\in\Sigma^\omega_t$ there exists a solution $(\xi,Q)$ for the \MSP\ such that
  $$
    Q\big[\Sigma^{\bar\omega}_{t+1}\big]>0,\quad 
    \xi_t(\omega)=y
    \quad\text{and}\quad
    \xi_{t+1}(\bar\omega)\in\ri\cT_{t+1}(\bar\omega).
  $$ 
  Indeed, by the definition of $\cT_{t+1}(\bar\omega)$, there exists a solution $(\xi',Q')$ to \MSP\ such that $Q'[\Sigma^{\bar\omega}_{t+1}]>0$ and $\xi_{t+1}(\bar\omega)\in\ri\cT_{t+1}(\bar\omega)$. Notice that $\xi'_t(\omega)\in\cT_t(\omega)$. Since $y\in \ri\cT_t(\omega)$, there exists a solution $(\xi'',Q'')$ to the \MSP\ and a $\lambda\in(0,1)$ such that $y = \lambda\xi'_t(\omega)+(1-\lambda)\xi''_t(\omega)$. Use the construction of Lemma~\ref{lem::convexity}, i.e. the argument above this proof showing convexity of $\cT_{t+1}(\omega)$, to conclude.
  
  \textsc{Step~2:} Let $y\in \ri\cT_t(\omega)$. We claim that for each finite $\{\omega^1,\ldots,\omega^p\}\subset\Sigma^\omega_t$ there exists a solution $(\xi,Q)$ to the \MSP\ such that
  $$
    Q\big[\Sigma^{\omega^i}_{t+1}\big]>0,\quad 
    \xi_t(\omega)=y
    \quad\text{and}\quad
    \xi_{t+1}(\omega^i)\in\ri\cT_{t+1}
    \quad\forall i=1,\ldots,p.
  $$ 
  Indeed, use Step~1 to get $(\xi^i,Q^i)$, local solutions for the \MSP\ at $\omega^i$, respectively, each satisfying the conclusions of Step~1. Then any convex combination, in the sense of Lemma~\ref{lem::convexity}, i.e. the argument above this proof, will do.

\textsc{Step~3: (the induction step)} Assume that $\ri\cT_{t+1}(\omega)\subset W_{t+1}(\omega)$ for every $\omega\in\Omega$; in particular that $W_{t+1}(\omega)$ is nonempty for every $\omega$. Then
$$
  \ri\cT_t(\omega)\subset W^\flat_{t+1}(\omega)-C(\omega),
$$
in particular, $\ri\cT_t(\omega)\subset W_t(\omega)$.

Note that the later statement follows directly from the former. Indeed, since $\cT_t(\omega)\subset V_t(\omega)$ for every $\omega$, we get $\ri\cT_t(\omega)\subset V_t(\omega)\cap(W^\flat_{t+1}(\omega)-C(\omega))=W_t(\omega).$

To show the inclusion, we will use Lemma~\ref{lem::characterization}. So, fix an $\omega\in\Omega$ and let $y\in \ri\cT_t(\omega)$ be arbitrary. We want to show that there exists a selection $Y\in\cL(\cF_{t+1};W_{t+1})$ and a measure $Q\in\cP(\omega)$ with $Q[\Sigma^\omega_t]=1$ such that 
  $$ 
    Y(\bar\omega)\in\ri\cT_{t+1}(\bar{\omega})\quad \forall\bar{\omega}\in\supp Q
    \quad\mbox{and}\quad
    y \in \E_Q[Y]-C(\omega).
  $$
  To this aim, let $\{\omega^1,\ldots,\omega^p\}\subset\Sigma^\omega_t$ be such that
  $$
    \aff (\cT_{t+1})^\sharp(\omega) = \aff\bigcup_{i=1}^p\cT_{t+1}(\omega^i).
  $$
  Indeed, we are working in $\R^d$, hence this set always exists and can be chosen such that $p\leq d$. Choose a solution $(\xi,\bar Q)$ to the \MSP\, such that $\xi_{t+1}(\omega^i)\in\ri\cT_{t+1}(\omega^i)$ for every $i=1,\ldots,p$; this exists by Step~2. Denote $c=\E_{\bar Q}[\xi_{t+1}]-y\in C_t(\omega)$. 

We claim that
\begin{align}\label{eqn::key thing}
  y+c\in
  \ri\conv\bigcup_{\bar\omega\in\supp\bar Q}\cT_{t+1}(\bar\omega)=
  \ri\conv\bigcup_{\bar\omega\in\supp\bar Q}\ri\cT_{t+1}(\bar\omega).
\end{align}
The equality of the two sets in equation~\eqref{eqn::key thing} is clear, since both sets are convex, relatively open and have the same closure. To see why this is enough, note that every element of the right hand side in equation~\eqref{eqn::key thing} may be written as
$$
  \sum_{\bar\omega\in\supp\bar Q}\lambda^{\bar\omega}\psi^{\bar\omega}\qquad\mbox{with }\,\lambda^{\bar\omega}\in(0,1),\ \psi^{\bar\omega}\in\ri\cT_{t+1}(\bar\omega)\ \forall\bar\omega\in\supp\bar Q. 
$$
Let $Y$ be any selection of $W_{t+1}$ and modify it on $\bigcup_{\bar\omega\in\supp\bar Q}\Sigma^{\bar\omega}_{t+1}$ as follows: $Y(\bar\omega) = \psi^{\bar\omega}$ for all $\bar\omega\in\supp\bar Q$ and define the measure $Q = \sum_{\bar\omega\in\supp\bar Q}\lambda^{\bar\omega}\delta_{\bar\omega}$.

It remains to prove~\eqref{eqn::key thing}. Start by showing that the set
$$
  A\coloneqq\bigg\{\sum_{i=1}^p \lambda^i\zeta^i\,\bigg|\,\zeta^i\in\ri\cT_{t+1}(\omega^i),\ \lambda^i\in(0,1),\ \lambda^1+\cdots+\lambda^p = 1\bigg\}
$$
is relatively open; it is of maximal dimension by definition of $\{\omega^1,\ldots,\omega^p\}$. Choose a maximal affinely independent set $\{x_1,\ldots,x_\ell\}\subset\bigcup_{i=1}^p\ri\cT_{t+1}(\omega^i)$. Then
$$
  A\supseteq B\coloneqq\bigg\{\sum_{i=1}^\ell \lambda^ix_i\,\bigg|\,\lambda^i\in(0,1),\ \lambda^1+\cdots+\lambda^p = 1\bigg\},
$$
and the set $B$ is relatively open and of maximal dimension. Then
$$
  A = \bigcup \lambda B + (1-\lambda)\sum_{i=1}^p \lambda^i\zeta^i;
$$
the union is over $\lambda\in(0,1)$ and the sum over $\zeta^i\in\ri\cT_{t+1}(\omega^i)$ and $\lambda^i\in(0,1)$ for every $i$ with $\lambda^1+\cdots+\lambda^p = 1$. This proves that $A$ is relatively open. Coming back to equation~\eqref{eqn::key thing}, we note that
\begin{align*}
  y+c 
  = 
  \E_{\bar Q}[\xi_{t+1}] 
  &= 
  \sum_{i=1}^p\bar Q[\Sigma^{\omega^i}_{t+1}]\xi_{t+1}(\omega^i) + \sum_{\substack{\bar\omega\in\supp\bar Q\\\bar\omega\not\in\{\omega^1,\ldots,\omega^p\}}}\bar Q[\Sigma^{\bar\omega}_{t+1}]\xi_{t+1}(\bar\omega)\\[1ex]
  &\in (\bar Q[\Sigma^{\omega^1}_{t+1}]+\cdots+\bar Q[\Sigma^{\omega^p}_{t+1}])A + \sum_{\substack{\bar\omega\in\supp\bar Q\\\bar\omega\not\in\{\omega^1,\ldots,\omega^p\}}}\bar Q[\Sigma^{\bar\omega}_{t+1}]\xi_{t+1}(\bar\omega)\\[1ex]
  &\subset
  \ri\conv\bigcup_{\bar\omega\in\supp \bar Q}\cT_{t+1}(\bar\omega).
\end{align*}
This establishes~\eqref{eqn::key thing}.

\textsc{Step~4:} We finish the proof by noticing that $\ri\cT_T(\omega)\subset W_T(\omega)$ for every $\omega$ and using Step~3 repeatedly.
\end{proof}

\section{Applications to problems of mathematical finance}
In this section we describe the connection between the martingale selection problem and the theory of arbitrage in various types of financial markets. We will provide the examples in increasing order of complexity. In what follows we always suppose that a Polish space $\Omega$ is given, it is endowed with its Borel sigma-algebra $\cB(\Omega)$, and that the trading dates are specified by $\cI\coloneqq\{0,\ldots,T\}$ with $T\in\N$ fixed. Furthermore, we assume from now on that all correspondences have conical values, unless explicitly stated otherwise.

\subsection{The frictionless market model}
The first example is that of a frictionless market model, described by a process $S\coloneqq(S_t)_{t\in\cI}$ with $S_t\colon(\Omega,\cB(\Omega))\mapsto(\R^{d},\cB(\R^{d}))$ for every $t\in\cI$. In addition to $S$, the agent also holds a position in the riskless asset, which we denote by $B=(B_t)_{t\in\cI}$. We assume that $B_t(\omega)=1$ for all $\omega,\,t$. Consider the filtration $\F:=(\cF_t)_{t\in\cI}$ as constructed in Section~\ref{Section: setup}. Positions in the risky asset are subject to constraints, modeled by a sequence $A\coloneqq(A_t)_{t\in\cI}$, with every $A_t\colon\Omega\rightrightarrows\R^d$ an $\cF_t$-measurable correspondence with convex, closed and conical values. More precisely, admissible strategies, i.e. positions in the risky asset, are
$$
  \cH_A = \big\{(h_t)_{t\in\cI} \,\big|\,h_t\in\cL(\cF_t;A_t)\ \forall t\in \cI\big\}.
$$
Clearly, $\cH_A$ is a convex cone. Position in the riskless asset, which we denote by $h^0=(h^0_t)_{t\in\cI}$ can be determined by the self-financing condition: at time $t\in\cI$ a change in the holdings in the risky asset need to be financed by a change in position in the riskless
$
  h^0_t-h^0_{t-1}=-\langle S_t,h_t-h_{t-1}\rangle 
$
with the convention that $h_{-1}=0$ and $h^0_{-1}\in\R$ is the initial capital. The value of a strategy $h\in\cH_A$ is given by
$$
  \mathcal{V}_T(h) = h^0_{-1}+\sum_{t=0}^{T-1} \big\langle h_t,S_{t+1}-S_t\big\rangle.
$$
Note that $\mathcal{V}_T(h)=h^0_T$ with the assumption $h_T=0$. We say that the market model is arbitrage-free if for every $h\in\cH_A$ with zero initial capital
$$
  \mathcal{V}_T(h) \geq 0 \ \ \forall\omega\in\Omega \qquad \Longrightarrow \qquad \mathcal{V}_T(h) = 0 \ \ \forall\omega\in\Omega.
$$
\begin{theorem}\label{thm:FTAPfrictionless}
The market model given by $(S_t)_{t\in\cI}$, $(B_t)_{t\in\cI}$ and $(A_t)_{t\in\cI}$ is arbitrage free if and only if for every $\bar{\omega}\in\Omega$ there exists a finite support measure $P\in\cP(\bar\omega)$ such that  
\begin{align}\label{eqn:frictionless}
  E_P[S_{t+1}-S_t\,|\,\cF_t]\in -A^\ast_t
  \qquad P\textup{-a.s.} \mbox{ for all }t.
\end{align}
\end{theorem}
\begin{proof}
First we show the `if' part of the statement. Let $h\in\cH_A$ be such that $h^0_{-1}=0$ and $\mathcal{V}_T(h)(\omega)\geq0$ for every $\omega\in\Omega$. By the statement of the theorem, for every $\bar\omega\in\Omega$ there exists a measure $P\in\cP(\bar\omega)$, such that~\eqref{eqn:frictionless} is satisfied. However, equation~\eqref{eqn:frictionless} implies that $E_P[\mathcal{V}_T(h)]=\sum_{t=0}^{T-1} E_P[\big\langle h_t,S_{t+1}-S_t\big\rangle]\leq0$, by the definition of the polar cone and the admissibility of $h$. Since $\mathcal V_T(h)$ is non-negative, it follows  that $\mathcal V_T(h)(\omega)=0$ for all $\omega\in\supp(P)$ and, in particular, for $\bar{\omega}$. Since $\bar\omega\in\Omega$ was arbitrary, the thesis follows.

We now show the `only if' part. Assume that the martingale selection problem $(V,C)$, given by
$$
  V_t = \ri\cone(1,S_t)\qquad\mbox{and}\qquad C_t=-(\R\times A_t)^\ast=-\{0\}\times A^\ast_t,
$$
is solvable. Let $\bar\omega\in\Omega$ be arbitrary and denote by $((y,\xi),Q)$ any local solution at $\bar\omega$. Then $\xi_t=y_tS_t$ for all $t$, by the definition of $V$. By the definition of $C$, the process $y$ is a martingale. Using $y$ as the density process we obtained the desired measure $P$.

It remains to show that the no-arbitrage condition implies the solvability of the \MSP. By Theorem~\ref{thm: main1}, we need to show that the correspondences $W_t$, defined in~\eqref{eq: iteration}, are nonempty for all $t,\,\omega$. We argue by contradiction: Let $t\in\cI$ be the largest index for which there exists an $\omega\in\Omega$ such that $W_t(\omega)=\varnothing$. This means, in particular, that $V_t(\omega)$ and $W^\flat_{t+1}(\omega)-C_t(\omega)$ are disjoint convex cones. Hence, there exists $z\in\R^{1+d}$ such that
\begin{align}\label{eqn:separation}
  \langle x,z\rangle\leq0\leq\langle y,z\rangle\quad\forall x\in V_t(\omega)\mbox{ and }y\in W^\flat_{t+1}(\omega)-C_t(\omega).
\end{align}
Moreover, $z$ can be chosen such that $0<\langle \bar y,z\rangle$ for some $\bar y\in W^\flat_{t+1}(\omega)$. To understand this separation, write $z=(z_0,\bar z)$, where $z_0\in\R$ and $\bar z\in\R^d$. First, a simple inspection of~\eqref{eqn:separation} yields $z\in(W^\flat_{t+1}(\omega)-C_t(\omega))^\ast=(W^\flat_{t+1}(\omega))^\ast\cap-C_t^\ast(\omega)$, i.e. $\bar z\in A_t(\omega)$. It is also easy to see that $\bar z\not=0$. Indeed, by the maximality of $t$, $W_{t+1}(\omega)\neq\varnothing$ for any $\omega$ and, since $V_{t+1}(\omega)$ is a ray in $\R^{d+1}$,~\eqref{eq: iteration} implies that $W_{t+1}(\omega)=V_{t+1}(\omega)$ for all $\omega$. If therefore, $\bar z=0$, the vector $z$ would not separate the two sets in~\eqref{eqn:separation}, since the first component of both sets is precisely the interval $(0,\infty)$. In addition,~\eqref{eqn:separation} and the fact that $C_t(\omega)$ is a cone, imply that $z\in(W^\sharp_{t+1})^*(\omega)$ and $z\in -V^*_t(\omega)$.
We thus obtain that, for every $\bar\omega\in\Sigma^\omega_t$,
$$
  0\leq
  \langle (1,S_{t+1}(\bar\omega)),(z_0,\bar z)\rangle -
  \langle (1,S_t(\omega)),(z_0,\bar z)\rangle
  =
  \langle S_{t+1}(\bar\omega)-S_t(\omega),\bar z\rangle
  =
  \mathcal{V}_T(h)(\bar\omega)
$$
where $h\in\cH_A$ is the strategy defined by $h_s=0$ for $s=\cI\backslash\{t\}$ and $h_t = \bar z\one_{\Sigma^\omega_t}$. Since there exists a $\bar y\in W^\sharp_{t+1}(\omega)$ such that the inequality is strict in~\eqref{eqn:separation}, this is an arbitrage strategy.
\end{proof}

\begin{remark}
  In the particular case of short-selling constraint, i.e. $A_t=\R_+^d$, we get that the market model is arbitrage free if and only if there exists a finite support measure $Q$, such that every component of $S$ is a $Q$ super-martingale.
\end{remark}

\begin{remark}
  In the course of the proof we have established that if a frictionless market model with portfolio constraints admits arbitrage, it also admits a one-step arbitrage. This result is well known in the theory of no-arbitrage market models without portfolio constraints.
\end{remark}
\begin{remark}
The notion of arbitrage considered here is a particular case of arbitrage \emph{de la classe} $\mathcal{S}$ introduced 
in \cite{BFM16} and called \emph{1p-arbitrage}.
In the proof of Theorem \ref{eqn:frictionless} we showed that 1p-arbitrages can be constructed as separators of the sets $V_t(\omega)$ and $W^\flat_{t+1}(\omega)-C_t(\omega)$ (for those $\Sigma^\omega_t$ for which they are disjoint). We can construct a 1p-arbitrage which is non-zero on every such level set by taking $z:=\sum_{n=1}^{\infty} \frac{1}{2^n|z_n|} z_n$ where $\{z_n\}_{n\in\N}$ is a Castaing representation of $(V_t-(W^\flat_{t+1}-C_t))^*$.
This is called a \emph{standard separator} in \cite{BFHMO16} and it is instrumental in deriving versions of the FTAP for arbitrages de la classe $\mathcal{S}$.
The same analysis could be replicated here with minor modifications.
\end{remark}

%%%%%%%%%%%%%%%%%%%%%%%%%%%%%%%%%%%%%%%%%%%%%%%%%%%%%%%%%%%%%%%%%
% - % - % - % - % - % - % - % - % - % - % - % - % - % - % - % - %
\subsection{Kabanov's model of currency markets}\label{section Kabanov}
 The financial market is fully described by a discrete-time process $K\coloneqq(K_t)_{t\in\cI}$, where every $K_t$ is a $\cB(\Omega)$-measurable correspondence whose values are closed cones in $\R^d$. We call the set $K_t$ a \emph{solvency cone} and its elements are portfolio compositions that can be liquidated, i.e. for which one can find a counterparty to take it at zero cost. We consider $\F\coloneqq(\cF_t)_{t\in\cI}$ as defined in Section~\ref{Section: setup}.

An adapted process $h=(h_t)_{t\in\cI}$ is called \emph{a self-financing strategy} if its increments can all be achieved at zero cost, i.e. 
%by subtracting from the portfolio $h_{t-1}$ a solvent position $h_{t-1}-h_t$ at time $t$. We write this as
$
  h_{t-1}-h_t\in K_t
$
with the convention that $h_{-1}=0$. We can also sophisticate the original Kabanov's model and introduce restrictions to the class of portfolios the trader is allowed to hold. We do this by introducing a conical constraints set $A\coloneqq(A_t)_{t\in\cI}$ where the correspondence $A_t\colon\Omega\rightrightarrows\R^d$ is $\cF_t$-measurable. Denote by
$$
  \cH_{K,A}\coloneqq\big\{(h_t)_{h\in\cI}\,\big|\,h_t\in\cL(\cF_t;A_t),\ h\mbox{ is self-financing}\big\},
$$
the class of admissible strategies.

 In this section we will need the following condition.
\begin{assumption}\label{ass: EF} For every $t\in\cI$ and every $\omega\in\Omega$
 $$
   \inter K_t(\omega)\supseteq\R^d_+\backslash\{0\}\quad\text{ and }\quad A_t(\omega)\cap \R^d_+\setminus\{0\}\neq\varnothing.
 $$
% moreover, $A_t(\omega)\subset A_{t+1}(\omega)$.
\end{assumption}

The first condition states that every non-negative position in the market is considered solvent, stated differently, one can freely dispose of assets. The second condition says that there exists at least one non-negative position which is allowed. 

 We next introduce some concepts of arbitrage.
\begin{definition}\label{def: dominance}
\begin{enumerate}
\item[(a)] Let $K_1$, $K_2$ be closed cones. We say that cone $K_1$ dominates cone $K_2$ if 
  $$
    K_2\setminus (K_2\cap-K_2)\subset\ri K_1.
  $$
\item[(b)] We say that a market model $(\widehat K,A)$ dominates the market model $(K,A)$ if for every $t\in\cI$ and $\omega\in\Omega$ the cone $\widehat K_t(\omega)$ dominates the cone $K_t(\omega)$.
\end{enumerate}
\end{definition}

\begin{definition}
An admissible strategy $h\in\cH_{K,A}$ is an arbitrage strategy if $h_T\in\cL(\cF_T;\R^d_+)\setminus\{0\}$. Define two types of no-arbitrage condition
\begin{description}
 \item[\ \ $\NA^w$] weak no-arbitrage: $(K,A)$ admits no arbitrage strategies;
 \item[\ \ $\NA^r$] robust no-arbitrage: $(K,A)$ is dominated by $(\widehat K,A)$ satisfying $\NA^w$.
\end{description}
\end{definition}
In the above definition of an arbitrage strategy, it is assumed that $h_{-1}=0$ and that, before maturity, it can be liquidated into a portfolio with non-negative entries in every asset and strictly positive in some. Clearly this should not be allowed and it is excluded by the condition $\NA^w$. The stronger condition $\NA^r$, exclude the possibility that an arbitrary small reduction of the set of consistent price system, allows for some arbitrage strategies.

Before stating the main result of this section, we need to make an additional technical assumption.
\begin{assumption}\label{ass:relatively open}
  One of the following conditions holds:
\begin{enumerate}
  \item $A_t(\omega)=\R^d$ for every $t,\,\omega$; or
  \item $K_t(\omega)\cap-K_t(\omega)=\{0\}$ for every $t,\,\omega$.
\end{enumerate}
\end{assumption}

The second condition is known in the literature as efficient friction.
It models a situation where any trade in the market is subject to non-zero transaction costs.
From a technical point of view, we only need this condition in presence of portfolio constraints.

\begin{theorem}\label{thm: FTAP Kabanov}
Under Assumptions~\ref{ass: EF} and~\ref{ass:relatively open} robust no-arbitrage holds if and only if for every $\bar\omega\in\Omega$ there exists $P\in\cP(\bar\omega)$ and a process $\xi\coloneqq(\xi_t)_{t\in\cI}$ such that 
$$
  E_P[\xi_{t+1}-\xi_t|\cF_t]\in -A^\ast_t,\quad P\text{-a.s.}
$$
 and $\xi_t$ takes values in $\ri K^\ast_t$ for every $t\in\cI$.
\end{theorem}

This theorem is known in the financial mathematics literature as the fundamental theorem of asset pricing, in the case when $A_t(\omega)=\R^d$ for all $t,\omega$. Indeed, in that case $A^\ast=\{0\}$ and the pair $(\xi,Q)$ is called `consistent price system': the process $\xi$ is a $Q$-martingale which take values in the relative interior of the polar of the solvency cones; see~\citep{S04}.
A concrete example is with $K_t^\ast=\cone(\{1\}\times [b^1_t,a^1_t]\times\cdots\times[b^{d-1}_t,a^{d-1}_t])$. Here the first asset serves as a num\'eraire and the two processes $(a_t)_{t\in\cI}$ and $(b_t)_{t\in\cI}$ describe the ask and bid prices of the remaining assets with respect to the num\'eraire.

In the presence of constrains, $P$ is a fair pricing measure in the market defined by $\xi$, because the value process of any trading strategy $h$ is a super-martingale under $P$.
Indeed, since $h_t\in\mathcal{L}(\mathcal{F}_t;A_t)$ and $E_P[\xi_{t+1}-\xi_t|\cF_t]\in -A^\ast_t$, $P$-a.s., we have
\[
E_P[h_t\cdot(\xi_{t+1}-\xi_t)|\cF_t]=h_t\cdot E_P[\xi_{t+1}-\xi_t|\cF_t]\le 0,
\]
from which the super-martingale property of $(\sum_{t=0}^i h_t\cdot(\xi_{t+1}-\xi_t))_{i\in\cT}$ follows.

\medskip
As in the frictionless case, we will connect this problem with an appropriate martingale selection problem. In fact, the main theorem is already stated in the form of a $\MSP$, namely, we take
\begin{equation}
\label{eq:MSP for Kabanov}
  V_t(\omega) = \ri K_t^\ast(\omega)\quad\mbox{ and }\quad C_t(\omega) = -A^\ast_t(\omega).
\end{equation}
In the rest of the section we analyze the \MSP\ $(V,C)$.

\begin{lemma}\label{lem: firstpart_suffK}
  Under Assumptions~\ref{ass: EF} and~\ref{ass:relatively open}, solvability of the \MSP\ $(V,C)$ implies $\NA^w$.
\end{lemma} 
\begin{proof}
Let $h\in\cH_{K,A}$ be such that $h_T\in\cL(\cF_T;\R_+^d)$. Let $\bar\omega\in\Omega$ be arbitrary and denote by $(\xi,Q)$ a local solution at $\bar\omega$. Since $\xi_T$ has strictly positive entries, by Assumption~\ref{ass: EF}, also $\langle \xi_T(\omega),h_T(\omega)\rangle\geq0$ for every $\omega$.
This implies $E_Q[\langle \xi_T,h_T\rangle]\geq0$. 
Note that, for any $0\le t\le T-1$, $\xi_t=E_Q[\xi_{t+1}|\cF_t]-\zeta_t$ for some $\zeta_t\in-A^\ast_t$. By the tower property, $\xi_t=E_Q[\xi_T|\cF_t]-E_Q[\zeta_{T-1}+\cdots+\zeta_t|\cF_t]$. 
Summing up, we obtain
\begin{align*}
  E_Q[\langle \xi_T,h_T\rangle]
  &=
  \sum_{t=0}^T E_Q[\langle \xi_T,h_t-h_{t-1}\rangle]
  =%\\&=
  \sum_{t=0}^T E_Q[\langle E_Q[\xi_T|\cF_t],h_t-h_{t-1}\rangle]
  \\[3pt]&=
  \sum_{t=0}^T E_Q[\langle \xi_t,h_t-h_{t-1}\rangle]+\sum_{t=0}^{T-1} E_Q[\langle \zeta_t+\cdots+\zeta_{T-1},h_t-h_{t-1}\rangle]
  \\[3pt]&=
  \sum_{t=0}^T E_Q[\langle \xi_t,h_t-h_{t-1}\rangle]+\sum_{t=0}^{T-1} E_Q[\langle \zeta_t,h_t\rangle],
\end{align*}
Since $\xi_t\in K^\ast_t$ and $h_t-h_{t-1}\in-K_t$ by admissibility, the first sum is non-positive. 
Since $\zeta_t\in-A^\ast_t$ and $h_t\in A_t$ by admissibility, also the second sum is non-positive.
We conclude that $E_Q[\langle \xi_T,h_T\rangle]\le 0$. Since $\bar\omega\in\Omega$ was arbitrary, the claim follows.
\end{proof}

Let $(\widehat{K},A)$ be a market model with $\widehat{K}$ dominating $K$; we denote by $(\widehat{V},C)$ the \MSP\ defined as in equation~\eqref{eq:MSP for Kabanov} with $\widehat{K}$ replacing $K$.

\begin{lemma}\label{lem: dominated solvable} 
Suppose that $(V,C)$ is solvable.  Under Assumptions~\ref{ass: EF} and~\ref{ass:relatively open} there exists a market model $(\widehat{K},A)$ which dominates $(K,A)$ and such that $(\widehat{V},C)$ is solvable.
\end{lemma} 
\begin{proof} By Theorem \ref{thm: main1}, solvability of $(V,C)$ implies $W_t(\omega)\neq\varnothing$ for all $t,\,\omega$. By Assumption~\ref{ass:relatively open} and Lemma~\ref{lem: ri flat}, those are also relatively open and we use this fact in a crucial way in the proof.

\textsc{Step 1.} Define a sequence of relatively open cones $\{U_t^n\}_{n\in\N}$ for each $t\in\cI$, such that
\begin{enumerate}
  \item[i)]  $\overline{U_t^n}\setminus\{0\}\subset V_t$ for any $n\in\N$;
  \item[ii)] for any $x\in V_t$ there exists $\bar{n}$ such that $x\in U_t^{\bar{n}}$;
  \item[iii)] $U^n_t$ is $\cF_t$-measurable.
\end{enumerate}

Let $(\xi_t)_{t\in\cI}$ be a collection of measurable selections of $(V_t)_{t\in\cI}$ and define
$$
  U^n_t\coloneqq\cone\bigg( \frac{n}{1+n}\big(V_t\cap B_1(0)\big)+\frac1{1+n}\xi_t\bigg),
$$
where $B_1(0)$ denotes a closed unit ball around $0$. It is easy to see that all the three properties are satisfied. Furthermore, $U^n_t(\omega)\subset U^{n+1}_t(\omega)$ for all $n,\,t,\,\omega$.

\textsc{Step 2.}
We will define an adapted sequence $(n_s)_{s\in\cI}$ such that the pair $(\widehat V,C)$, where $\widehat V_s\coloneqq U_s^{n_s}$ for all $s$, defines a solvable \MSP; we use the characterization of solvability given in Theorem \ref{thm: main1}. Set $n_s^T=1$ identically for all $s\in\cI$. Assume we defined $(n_s^{t+1})_{s\in\cI}$ for $t\in\cI\backslash\{T\}$ and we proceed to define the sequence $(n_s^t)_{s\in\cI}$. Consider the \MSP s $(V^{t+1,j},C)$, given by $V^{t+1,j}_s= U^{n_s^{t+1}+j}_s$ for all $j\in\N$ and $s\in\cI$. The corresponding sequences $W^{t+1,j}_s$ are given by
\begin{align*}
  W^{t+1,j}_s = V^{t+1,j}_s\cap \big((W^{t+1,j}_{s+1})^\flat-C_s\big)
  \qquad \mbox{ for }s=T-1,\ldots, 0.
\end{align*}
By the inductive assumption $W^{t+1,j}_s(\omega)\not=\varnothing$ for all $j\in\N$ and $s=t+1,\ldots,T$. Since $V^{t+1,j}_s$ is an increasing sequence whose union equals $V_s$, an analogous observation also holds for $W^{t+1,j}_s$ and $(W^{t+1,j}_{s})^\flat$ with respect to $W_s$ and $W_{s}^\flat$ . Since the sets $V_t$ and $W_t$ are relatively open, the following random variable
$$
  \zeta = \inf\{j\in\N\,|\,W^{t+1,j}_t\not=\varnothing\}
$$
has finite values. Define $n^t_s=n^{t+1}_s+\zeta\one_{s\geq t}$; it is clear that the sequence $(n^t_s)$ is adapted. Finish the proof by setting $n_s=n^0_s$ for all $s\in\cI$.
\end{proof}

\begin{corollary}\label{cor: suff Kabanov}
Under Assumptions~\ref{ass: EF} and~\ref{ass:relatively open}, solvability of the \MSP\ $(V,C)$ implies $\NA^r$.
\end{corollary}
\begin{proof}
	It follows directly from Lemma \ref{lem: firstpart_suffK} and Lemma \ref{lem: dominated solvable}.
\end{proof}

To study the reverse implication, let us start by defining a sequence of measurable correspondences: set $w_T=V_T$ and
\begin{align}\label{eq: msp ri}
  w_t&=V_t\cap \ri\big(w^\sharp_{t+1}-C_t\big)
  \qquad \text{ for }t=T-1,\ldots, 0.
\end{align}
\begin{remark}\label{rmk: msp ri}
  It is clear that $w_t\subset W_t$ for every $t\in\cI$. Argue by induction: for $t=T$ it is true by definition, if $w_{t+1}\subset W_{t+1}$, then also 
$$
  \cl(w^\sharp_{t+1}-C_t)=\cl(w^\flat_{t+1}-C_t)\subset\cl(W^\flat_{t+1}-C_t)
$$
which yields the same inclusion for $t$.
\end{remark}

\begin{lemma}\label{lem: polar Y}
Let $t\in\cI\backslash\{T\}$ and $\omega\in\Omega$ such that $w_{t}(\omega)\neq\varnothing$. Then
$$
w_{t}^*(\omega)= K_{t}(\omega)+((w^\sharp_{t+1})^*\cap- C^\ast_{t})(\omega) .
$$
\end{lemma}
\begin{proof}
Since $\omega\in\Omega$ is fixed, we omit, in what follows, the dependence on $\omega$. Observe that
\begin{align*}
    w_{t}^* 
    &=
    \big(V_{t} \cap \ri(w^\sharp_{t+1}-C_{t}) \big)^\ast=
    %\\[1ex]&=
   % \cl{\conv\big(V_{t}^* \cup (w^\sharp_{t+1}-C_{t})^* \big)}
    %\\[1ex]&=
%    \cl{(V_{t}^* +(w^\sharp_{t+1}-C_{t})^* )}\
%     \\[1ex]&=
     V_{t}^* +(w^\sharp_{t+1}-C_{t})^* =
     %\\[1ex]&=
    K_{t} +((w^\sharp_{t+1})^*\cap- C^\ast_{t}) \ .
  \end{align*}
Indeed, the first equality is simply the definition of $w_t$. The assumption $w_{t}\neq\varnothing$ implies that $V_{t}\cap \ri(w^\sharp_{t+1}-C_t)\neq\varnothing$. 
%The second equality follows from the two sets being cones. The assumption $w_{t}\neq\varnothing$ guarantees that the sets $V_{t}$ and $\ri(w^\flat_{t+1}-C_t)$ satisfy 
Thus, the assumption of Corollary~16.4.2 in~\citep{R70} is satisfied and the second equality follows. The last one follows from Corollary~16.4.2 in~\citep{R70} and the definition of $V_t$.
\end{proof}

\begin{lemma} \label{lem: polar sharp}
Let $t\in\cI\backslash\{T\}$ and $\bar \omega\in\Omega$. Assume that $w_{u}(\omega)\neq\varnothing$, for every $t\le u\leq T$ and $\omega\in\Omega$. If $z_t\in (w^\sharp_{t+1}-C_t)^*(\bar\omega)$ then
\begin{align}\label{eq: decomposition}
   z_t=k_{t+1}+\ldots+k_T\quad\text{ with } k_u\in\cL(\cF_u;K_{u})\quad \text{ for }u=t+1,\ldots, T.
\end{align}
If, in addition, $z_t\notin -(w^\sharp_{t+1}-C_t)^*(\bar\omega)$, then there exists an $\hat{\omega}\in\Sigma_t^{\bar \omega}$ and $t+1\leq \hat{u}\leq T$ such that $k_{\tilde{u}}(\hat\omega)\in (K_{\hat{u}}\setminus (K_{\hat{u}}\cap- K_{\hat u}))(\hat{\omega})$.
\end{lemma}

\begin{proof} 
Fix $t\in\cI\backslash\{T\}$ and $\bar \omega\in\Omega$. Since $(w^\sharp_{t+1}-C_t)^*=(w^\sharp_{t+1})^\ast\cap-C_t^*$, we have $z_t\in (w^\sharp_{t+1})^*$. From Lemma~5.102 in \citep{Aliprantis}
$$
  (w^\sharp_{t+1})^{*}(\bar\omega)=\bigcap_{\omega \in\Sigma_t^{\bar{\omega}}}w_{t+1}^{*}(\omega).
$$
In particular, $z_t$ is a selection of $w^\ast_{t+1}$. Fix $\omega\in\Sigma_t^{\bar \omega}$.  Since $w_{t+1}(\omega)\neq\varnothing$, Lemma~\ref{lem: polar Y} and Lemma~\ref{lem:appendix sum} imply that 
\begin{align}\label{eq: decomposition proof}
    z_t=k_{t+1}+z_{t+1}
\end{align}
with $k_{t+1}\in \cL(\cF_{t+1};K_{t+1})$ and $z_{t+1}\in\cL(\cF_{t+1};((w^\sharp_{t+2})^*\cap-C_{t+1}^\ast))$. If $z_{t+1}=0$, we set $k_u=0$ for $u=t+2,\ldots T$ and the representation \eqref{eq: decomposition} follows. If $z_{t+1}\neq 0$, we iterate the procedure on $z_{t+1}$ up to time $T-1$. Recalling that $w_T=V_T$, the representation \eqref{eq: decomposition} follows.

As for the second assertion, observe that if for all $k_{t+1}$, for which the decomposition~\eqref{eq: decomposition proof} holds, $k_{t+1}\in K_{t+1}\cap-K_{t+1}$, then $z_t\notin -(w^\sharp_{t+1}-C_t)^*(\bar\omega)$ implies $z_{t+1}\notin  -(w^\sharp_{t+2}-C_{t+1})^*(\hat{\omega})$ for some $\hat{\omega}\in\Sigma_{t}^{\bar \omega}$ (in particular $z_{t+1}\neq 0$). Iterating the procedure on $z_{t+1}$ up to time $T-1$ and recalling that $w_T=V_T$, the thesis follows.
\end{proof}

A useful tool for constructing arbitrage strategies is given by the following.

\begin{lemma}\label{lem: arb sum}
Under Assumption~\ref{ass: EF}, assume that there exists an admissible strategy $h\in\cH_{K,A}$ such that $h_T=0=\sum_{u=0}^{T}-k_u$ with $k_t(\omega)\in \inter K_t(\omega)$ for some $t\in\cI$ and $\omega\in\Omega$. Then there exists an arbitrage strategy $\hat{h}\in\cH_{K,A}$.
\end{lemma}

\begin{proof}
%For any $s\in\cI$, denote by $\cH_{K,A}(s)\subset\cH_{K,A}$ the set of strategies $h$ such that $h_T=0$ and $\exists\bar{\omega}\in\Omega$ for which $k_s(\bar\omega)\in\inter K_s(\bar{\omega})$. 
%Let $t:=\max\{s\in\cI\mid \cH_{K,A}(s)\neq\varnothing\}$.
%By assumption $\bar{t}\le t\le T$.
%Since $k_t(\bar\omega)\in \inter K_{t}(\bar \omega)$, there exists $x\in A_{t}(\bar \omega)\cap \R^d_+\setminus\{0\}$ (which is non-empty by Assumption~\ref{ass: EF}) such that $k_t(\bar\omega)- x\in \inter K_{t}(\bar \omega)$. 
Let $h\in\cH_{K,A}$ be a strategy as in the lemma.

\textsc{Case 1:} if $t=T$, $h$ can be modified to an arbitrage strategy. Indeed, since $k_t(\omega)\in\inter K_t(\omega)$, there exists $x\in A_t(\omega)\cap\R_+^d\backslash\{0\}$ such that $k_t(\omega)-x\in\inter K_t(\omega)$. The strategy $\hat h = h+x\one_{T}\one_{\Sigma_T^{{\omega}}}$ is in $\cH_{K,A}$ and satisfies $\hat h_T = x\one_{\Sigma_T^{{\omega}}}$. 

\textsc{Case 2:} assume $t<T$; we want to show that there exists a strategy satisfying the condition with a higher time index. As in \textsc{Case 1}, there exists $x\in A_t(\omega)\cap\R_+^d\backslash\{0\}$ such that $k_t(\omega)-x\in\inter K_t(\omega)$. Let us show that the strategy $\hat h = h+x\one_t\one_{\Sigma_T^{{\omega}}}$ is in $\cH_{K,A}$. It is clear that $\hat h_s\in\cL(\cF_s,A_s)$ for all $s\in\cI\backslash\{t\}$; it is also true for $s=t$, since $\hat h_t = h_t + x\one_{\Sigma_T^{{\omega}}}$ and $A_t(\omega)$ is a convex cone. To check that the increments are contained in the correct sets, note that $\hat k_s = k_s$ for $s\in\cI\backslash\{t,t+1\}$. Furthermore,  $\hat k_t = k_t - x\one_{\Sigma_T^{{\omega}}}$ and $\hat k_{t+1} = k_{t+1} + x\one_{\Sigma_T^{{\omega}}}$; the first element is a selection of $K_t$ by construction, the second one because $x\in\R^d_+\backslash\{0\}\subset\inter K_{t+1}$ by Assumption~\ref{ass: EF}. The latter inclusion shows that the strategy $\hat h$ satisfies: $\exists\bar\omega$ such that $\hat k_{t+1}(\bar\omega)\in\inter K_{t+1}(\bar\omega)$.

This finishes the proof, since if there exists a strategy $h\in\cH_{K,A}$ such that $h_T=0=\sum_{u=0}^{T}-k_u$ with $k_t(\omega)\in \inter K_t(\omega)$ for some $t\in\cI$ and $\omega\in\Omega$, then one sees that there also exists a strategy $h'$ satisfying the same property for $t=T$, by applying \textsc{Case 2} repeatedly. From \textsc{Case 1} there exists an arbitrage strategy.
%
%\textsc{Case 1:} if $t=T$, the strategy  $\hat h_s=h_s+ x\one_{\{s=T\}}\one_{\Sigma_{T}^{\bar \omega}}$ is admissible, since $A_T$ is a convex cone and $k_T(\bar\omega)- x\in \inter K_{T}(\bar \omega)$. Moreover, it satisfies $\hat {h}_T=h_T+ x\one_{\Sigma_{T}^{\bar \omega}}= x\one_{\Sigma_{T}^{\bar \omega}}$, which is non negative and non zero on $\Sigma_{t}^{\bar{\omega}}$. Thus, $\hat{h}$ is an arbitrage.
%
%\textsc{Case 2:} if $t<T$, we have a contradiction.
%Define $\hat k_s = k_s$ for $s\in\cI\backslash\{t,t+1\}$, $\hat k_t\coloneqq k_t- x\one_{\Sigma^{\bar\omega}_t}$ and $\hat k_{t+1} = k_{t+1} +  x\one_{\Sigma_t^{\bar{\omega}}}$. Note that $\hat k_{t}$ is a selection of $\inter K_{t}$ from the above and $\hat k_{t+1}$ is a selection of $\inter K_{t+1}$ since $x\in\R^d_+\backslash\{0\}\subset\inter K_{t+1}$  by Assumption~\ref{ass: EF}.
%We show that the strategy $\hat{h}$ with $\hat{h}_s:=\sum_{u=0}^s-\hat{k}_u$ belongs to $\cH_{K,A}(t+1)$, which contradicts the maximality of $t$.
%To see this, observe that $\hat{h}_s=h_s$ for any $s\in\cI\backslash\{t\}$, thus, $\hat{h}_T=0$ and $\hat h_s\in\cL(\cF_s,A_s)$. 
%Moreover, the same is true for $s=t$, since $\hat h_t = h_t + x\one_{\Sigma_T^{\bar{\omega}}}$ and $A_t(\bar\omega)$ is a convex cone.
\end{proof}

We are now ready to finish the proof of Theorem \ref{thm: FTAP Kabanov}. Note that the proof only requires Assumption \ref{ass: EF}.

\begin{lemma}\label{lem: necess Kabanov}
Under Assumption \ref{ass: EF}, $\NA^r$ implies that the \MSP\ $(V,C)$ is solvable.
\end{lemma} 
 
\begin{proof} 
We argue by contradiction. Assume that the \MSP\ is not solvable; by Theorem \ref{thm: main1}, $W_t(\omega)=\varnothing$ for some $t\in\cI$ and $\omega\in\Omega$. Since $(w_t)_{t\in\cI}$, defined in~\eqref{eq: msp ri}, satisfies $w_t\subset W_t$ for every $t\in\cI$ (see Remark \ref{rmk: msp ri}), also $w_t(\omega)=\varnothing$. Choose the largest index $t\in\cI$ for which there exists an $\bar{\omega}$ such that $w_t(\bar{\omega})=\varnothing$. Since $w_T=V_T$, this is well defined and $t\leq T-1$. The assumption $w_t(\bar{\omega})=\varnothing$ implies that $V_t(\bar{\omega})$ and the relative interior of $w_{t+1}^\sharp(\bar{\omega})-C_t(\bar\omega)$ are disjoint. We deduce the existence of a vector $z\in\R^d\backslash\{0\}$ such that
  \begin{align}\label{eq: separator}
  \langle x,z\rangle\le 0 \le \langle y,z\rangle\quad\forall x\in  V_t(\bar\omega)\text{ and }y\in  (w_{t+1}^\sharp-C_t)(\bar\omega).
  \end{align}
 This readily implies $-z\in V_t^*(\bar{\omega})=K_t(\bar{\omega})$, and $z\in(w_{t+1}^\sharp-C_t)^*(\bar\omega)$. In particular, $z\in-C_t^\ast(\bar\omega)=A_t(\bar\omega)$. Lemma \ref{lem: polar sharp}, implies that
\begin{align*}
  z=k_{t+1}+\ldots+k_T\quad\text{ with } k_u\in\cL(\cF_u;K_u).
\end{align*}
The vector $z$ in \eqref{eq: separator} can be chosen such that one of the conditions hold.
\begin{enumerate}
\item There exists a vector $\bar y\in w^\sharp_{t+1}(\bar \omega)$ such that $0<\langle \bar y,z\rangle$, which implies $z\notin -(w_{t+1}^\sharp-C_t)^*(\bar\omega)$. From Lemma \ref{lem: polar sharp}, there exist $\tilde{\omega}\in\Sigma_t^{\bar \omega}$ and $t+1\le {\tilde{u}}\le T$ such that $k_{\tilde{u}}\not\in (K_{\tilde{u}}\cap-K_{\tilde{u}})(\tilde{\omega})$.
\item There exists a vector $\bar x\in V_t(\bar\omega)$ such that $\langle \bar x,z\rangle<0$, which implies $z\not\in K_t(\bar\omega)\cap -K_t(\bar\omega)$. In this case set $\tilde{u}=t$ and $k_t\coloneqq-z\one_{\Sigma_t^{\bar\omega}}$.
\end{enumerate}
Consider an arbitrary process $\widehat{K}:=(\widehat{K}_u)_{u\in I}$ which dominates $(K_u)_{u\in\cI}$. Since $k_{\tilde{u}}\in(K_{\tilde{u}}\setminus (K_{\tilde{u}}\cap- K_{\tilde u}))(\tilde{\omega})
  \subset \inter \widehat{K}_{\tilde{u}}(\tilde{\omega})$, from Lemma \ref{lem: arb sum}, $\NA^r$ fails.
\end{proof}
    
\begin{proof}[Proof of Theorem \ref{thm: FTAP Kabanov}]
Necessity follows from Lemma \ref{lem: necess Kabanov}. Sufficiency follows from Corollary \ref{cor: suff Kabanov}.
\end{proof}

%%%%%%%%%%%%%%%%%%%%%%%%%%%%%%%%%%%%%%%%%%%%%%%%%%%%%%%%%
% - % - % - % - % - % - % - % - % - % - % - % - % - % - %
\subsection{Models of illiquidity}
A quite general form of discrete-time financial market model, including illiquid markets, is proposed in~\citep{P11} in a probabilistic framework. The main modeling tool is a \emph{cost process} $S=(S_t)_{t=0}^T$ which satisfies, for every $t=0,\ldots,T$, the following properties
\begin{itemize}
\item $S_t\colon\Omega\times\R^d\rightarrow \overline{\R}$ is Borel-measurable;
\item for every $\omega\in\Omega$ fixed the map $S_t(\omega,\cdot)$ is convex, lower semi-continuous and $S_t(\omega,0)=0$.
\end{itemize}
These two properties, together with the assumption $\inter(\dom S_t )\neq \varnothing$, imply that $S_t$ is a normal integrand; see \citep{R}, Chapter 14. Normality, in turn, guarantees that the recession map of $S_t$ and evaluations are measurable.

We will present a somewhat simplified version of the model proposed in~\citep{P11}; our formulation is similar to~\citep{CR07}. The reader is referred to~\citep{P11} for a general modeling considerations of illiquidity. Our modification consists in assuming that, in addition to $S$, there exists a riskless asset $B=(B_t)_{t\in\cI}$, which we assume normalized to $B_t(\omega)=1$ for all $\omega,\,t$. Positions in the riskless asset we denote by $h^0=(h^0_t)_{t\in\cI}$. This allows us the following simple interpretation of the cost process: a change of position in the risky asset of $h_t-h_{t-1}$ at time $t\in\cI$ elicits a change of position in the riskless asset 
$$
   h^0_t \leq h^0_{t-1} -  S_t(h_t-h_{t-1})\qquad \forall t\in\cI
$$
with the convention that $h_{-1}=0$ and $h^0_{-1}\in\R$ is the initial capital. Few instances of this model are the following
\begin{enumerate}
\item frictionless markets: $S_t(\omega,x)=\langle x, s_t(\omega)\rangle$, for an $\R^d$-valued stochastic process $(s_t)_{t=0}^T$ representing the price of $d$ assets at time $t\in\cI$. 
\item bid-ask spreads: $S_t(\omega,x)=x\left(a_t\one_{x\geq 0}+ b_t\one_{x< 0}\right)$, where the processes $(a_t)_{t=0}^T$ and $(b_t)_{t=0}^T$ represent the bid and ask prices of a single asset.
\item non-linear transaction costs: $S_t(\omega,x)=s_t(\omega)\varphi(x)$ for a real-valued stochastic process $(s_t)_{t=0}^T$ and a strictly positive, increasing and convex function $\varphi$ representing the cost of illiquidity; see \citep{CR07}. 
\end{enumerate}

\begin{remark}\label{rmk: phisical}
In Section \ref{section Kabanov} a model with \emph{physical delivery} is considered. It asks for all the positions $X\in\cL(\cF;\R^d)$ that can be superhedged in the market with respect to partial relation given by the cone $\R^d_+$.
In this section we consider claims with \emph{cash delivery}.
\end{remark}

Trading restriction are introduced by means of a conical process $A$ of portfolio constraints, so that, the class of admissible strategies is given by
$$
  \cH_A\coloneqq\big\{(h_t)_{t\in\cI}\,\big|\,h_t\in\cL(\cF_t;A_t),\ \forall \omega,\,t\in\cI,\ h_T=0\big\}.
$$
The value of a strategy $h\in\cH_A$ is given by
$$
\mathcal{V}_T(h) = h^0_{-1}-\sum_{t\in\cI} S_{t}(\omega, h_t-h_{t-1});
$$
remember the assumption $h_{-1}=0$. Note that $\mathcal{V}_T(h)=h^0_T$.

\begin{definition}
 A strategy $h\in\cH_A$ with zero initial capital $h^0_{-1}=0$ is called an \emph{arbitrage} if $\mathcal{V}_T(h) \geq 0$ for any $\omega\in\Omega$ and is strictly positive for some $\bar\omega\in\Omega$.
 An arbitrage is called \textit{scalable} if $\alpha h$ is an arbitrage strategy for every $\alpha>0$.
\end{definition}
\begin{remark}
The no scalable arbitrage condition does not exclude strategies yielding positive gains at no risk. Nevertheless, these gains cannot be arbitrarily scaled. This is a conceptual difference between liquid and illiquid markets.
\end{remark}

For a convex, lower semi-continuous function $S_t(\omega,\cdot)$ with $S_t(\omega,0)=0$, the horizon function $S_t^\infty$ is given by 
$$
  S_t^\infty(\omega, x)\coloneqq\lim_{\alpha\rightarrow\infty}\frac1\alpha S_t(\omega,\alpha x)
  \qquad 
  \forall x\in\R^d,\ \omega\in\Omega.
$$
\begin{remark}
If $S_t$ is positively homogeneous $S_t^\infty$ coincides with $S_t$. More generally, when $S^\infty$ is pointwise finite, it represents the minimal positively homogeneous model whose cost process is greater or equal than $S$.
\end{remark}
By Exercise 14.54 in \citep{R}, if $S_t$ is a normal integrand the same is true for $S_t^\infty$; convexity is obviously preserved by the operation $(\cdot)^\infty$. By using Theorem 14.56 and Proposition 14.11 in \citep{R}, the following are Borel-measurable correspondences, for every $t\in\cI$
\begin{equation}\label{eq: def cone}
 V_t\coloneqq\ri\cone\big\{(1,v)\,\big|\, v\in\partial (S_t^\infty)(\cdot,0)\big\}. 
\end{equation}
Note that the set of scalable portfolio rebalancings is given by
\begin{align}\label{eq:defn of K}
  -K_t(\omega)
   =
     \big\{(\delta,\Delta)\in\R\times\R^d\,\big|\,\delta + S^\infty_t(\Delta)\leq0\big\} 
   =
     -V_t^\ast. 
\end{align}
We say that a cost process $\widehat{S}$ dominates $S$ if the corresponding $\widehat{K}$, as in~\eqref{eq:defn of K}, dominates $K$ in the sense of Definition~\ref{def: dominance}.

Let us introduce the definition of arbitrage.
\begin{definition}
Robust no (scalable) arbitrage holds if $S$ is dominated by $\widehat{S}$ and $\widehat S$ satisfies no (scalable) arbitrage. 
\end{definition}

Similarly as in Section \ref{section Kabanov} we require Assumptions~\ref{ass: EF} and~\ref{ass:relatively open} to hold for the associated market model $(K,A)$. Let us be more explicit on the assumptions. First, the assumption $\R^{d+1}_+\backslash\{0\}\subset\inter K_t(\omega)$ implies that the function $x\mapsto S_t(\omega,x)$ is strictly increasing for each $t,\omega$. Indeed, choose an arbitrary $\Delta\in\R^d_+$, $\Delta\not=0$. Then, since $(0,\Delta)\in\inter K_t$, there exists a $\delta<0$ such that $(\delta,\Delta)\in\inter K_t$. Going back to the cost process, this implies that $S^\infty_t(-\Delta)\leq\delta<0$, i.e. the cost process is strictly increasing with respect to relation induced by the cone $\R^d$. Finally, the assumption that $K_t\cap-K_t=\{0\}$ is easy to interpret and implies that for each $x\in\R^d\backslash\{0\}$ we have $S_t^\infty(x)>-S_t^\infty(-x)$.

The following is the main result of the section. 
\begin{theorem}\label{thm: pennanen}
 Under Assumptions~\ref{ass: EF} and~\ref{ass:relatively open} robust no scalable arbitrage holds if and only if for every $\bar\omega\in\Omega$ there exists $P\in\cP(\bar\omega)$ and a process $\xi\coloneqq(\xi_t)_{t\in\cI}$ such that 
$$
  E_P[\xi_{t+1}-\xi_t\,|\,\cF_t]\in -A^\ast_t\qquad P\text{-a.s.}
$$
and $\xi_t$ takes values in $ \ri K^\ast_t$, for every $t\in\cI$.
\end{theorem}
The financial interpretation of the pair $(\xi,P)$ is similar as the one after Theorem \ref{thm: FTAP Kabanov}, namely, $(\xi,P)$ defines an arbitrage-free frictionless price process which is compatible with the frictions considered for the market. Indeed, as in Theorem \ref{thm: FTAP Kabanov} the measure $P$ is fair in the sense that the value process of every trading strategy is a super-martingale under $P$. Moreover, the price process modeled by $\xi$ takes values in the range of prices that are observable in the market if agents trade with the market impact prescribed by $S$.

Once again, we relate the problem to the solution of the appropriate \MSP\ $(V,C)$ with $V$ as in \eqref{eq: def cone} and $C_t =- (\R\times A_t)^\ast$.

\begin{lemma}\label{lem: suff Penn}
If the \MSP\ $(V,C)$ is solvable then robust no scalable arbitrage holds.
\end{lemma}
\begin{proof}
%Suppose, by contradiction, that there exists a scalable arbitrage opportunity $h\in\cH_A$.
 Since $(V,C)$ is solvable, from Lemma \ref{lem: dominated solvable} there exists a dominating conical market $\widehat{K}$ such that the \MSP\ $(\widehat{K}^\ast,C)$ is solvable. From the definition of the correspondence $K$, equation~\eqref{eq:defn of K}, it is clear how it induces a market $\widehat{S}\coloneqq(\widehat{S}_t)_{t\in\cI}$ which, by construction, dominates $S$ and satisfies $\widehat{S}_t=\widehat{S}^\infty_t$. Assume that there exists $h$ such that 
$$
\mathcal{V}_T(h)(\omega) = -\sum_{t=0}^{T} \widehat{S}_{t}(\omega, h_t-h_{t-1})\ge 0 \qquad \forall\omega\in\Omega
$$
and strictly positive in some $\bar\omega\in\Omega$. Denote by $(\xi,Q)$ the local solution of $(\widehat{K}^\ast,C)$ at $\bar\omega$. From Theorem~8.30 in~\citep{R}
$$
  \widehat{S}_{t}(\omega,x)
  =
  \sup\big\{\langle v, x\rangle\,\big|\, v\in \partial \widehat{S}_{t}(\omega,0)\big\},\qquad \omega\in\Omega,\ t\in\cI,\ x\in\R^d.
$$
We now argue as in the frictionless case. From $\xi_t\in\cL(\cF_t;\ri \widehat{K}^\ast_t)$ and the particular form of the constraint correspondence $C_t$ we write $\xi_t = (z_t,y_t)$, where $z_t$ is a martingale, which is strictly positive by assumptions. We use it to change the measure and obtain $P$. By definition of the correspondence $V_t$, we get that $y_t/z_t$ is a selection of $\partial \widehat S_t(0)$. We, thus, have
\begin{align*}
 0 
 \leq 
   -\sum_{t=0}^{T} \widehat{S}_{t}(\omega, h_t-h_{t-1}) 
 \leq
   -\sum_{t=0}^{T}\Big\langle \frac{y_t}{z_t}, h_t-h_{t-1}\Big\rangle 
 =
    \sum_{t=0}^{T-1} \Big\langle h_t, \frac{y_{t+1}}{z_{t+1}}-\frac{y_t}{z_t},\Big\rangle.
\end{align*}
Since $\bar\omega\in\supp P$, by taking expectations with respect to $P$ we obtain $0<-\sum_{t=0}^{T} E_P[\widehat{S}_{t}(\cdot, h_t-h_{t-1})] \le 0$ which is clearly a contradiction.
\end{proof}

\begin{lemma}\label{lem: necess Penn}
Robust no scalable arbitrage implies solvability of the \MSP\ $(V,C)$.
\end{lemma}
\begin{proof}
 Suppose, by contradiction, that the $\MSP$ is not solvable. By Theorem \ref{thm: main1}, $W_t(\omega)$ is empty for some $t\in\cI$ and $\omega\in\Omega$. We argue as in the proof of Lemma \ref{lem: necess Kabanov}. First we observe that $w_t(\omega)$, defined in~\eqref{eq: msp ri}, is also empty. We next choose the largest index $t\in\cI$ for which there exists an $\bar{\omega}$ such that $w_t(\bar{\omega})=\varnothing$. We deduce the existence of a vector $z\in\R^d\backslash\{0\}$ such that
 \begin{align*}
 \langle x,z\rangle\le 0 \le \langle y,z\rangle\quad\forall x\in  V_t(\bar\omega)\text{ and }y\in  (w_{t+1}^\sharp-C_t)(\bar\omega).
 \end{align*}
 By the same argument as in the proof of Lemma \ref{lem: necess Kabanov} we have
 \begin{align}\label{eq: representation}
 z=k_{t+1}(\omega)+\ldots+k_T(\omega)\quad\text{ with } k_u(\omega)\in\cL(\cF_u;K_u).
 \end{align}
 Moreover, by setting $k_t\coloneqq-z\one_{\Sigma_t^{\bar\omega}}$, there exist $\tilde{\omega}\in\Sigma_t^{\bar \omega}$ and $t\le \tilde{u}\le T-1$ such that $k_{\tilde{u}}\in(K_{\tilde{u}}\setminus (K_{\tilde{u}}\cap- K_{\tilde u}))(\tilde{\omega})$.
 
 Consider an arbitrary process $\widehat{S}\coloneqq(\widehat{S}_t)_{t\in\cI}$ which dominates $S$. By definition, the corresponding $\widehat{K}$ as in \eqref{eq:defn of K} dominates $K$ and, in particular, $k_{\tilde{u}}\in\inter\widehat{K}_{\tilde{u}}(\tilde{\omega})$.
  Since $k_u\in\widehat{K}_{u}$ for every $t\le u\le T$, in particular, 
$$
  0
  \leq
  \langle (1,v), k_u\rangle
  = 
  k^0_u+\langle v, \bar k_u\rangle
  \qquad 
  \forall v\in \partial \widehat{S}^\infty_u(\omega,0).
$$
From Theorem~8.30 in~\citep{R}, we have
$$
  \widehat{S}^\infty_u(\omega,-\bar k_u)=\sup\big\{\langle v, -\bar k_u\rangle\,\big|\,v\in \partial \widehat{S}^\infty_u(\omega,0)\big\}\le k^0_u,
$$ 
and strictly negative for $\tilde{u}$ and $\bar{\omega}$. Recall that, by definition of horizon function, $$\widehat{S}_{u}(\omega,-\bar k_u)\le \widehat{S}^\infty_{u}(\omega,-\bar k_u).$$ Therefore, the self-financing strategy $h$ with $h_u=0$ for $u\le t-1$ and $h_u-h_{u-1}=-\bar k_u$ for $t\le u\le T$, satisfies
 $$
\mathcal{V}_T(h) = \sum_{u=t}^{T} -\widehat{S}_{u}(\omega, h_u-h_{u-1})\ge -\sum_{u=t}^{T} k^0_u(\omega)=0,
$$
 where the last equality follows from \eqref{eq: representation}. Since the above inequality is strict for $\tilde{\omega}$, we deduce that $h$ is an arbitrage. Since $\widehat{K}$ is conical, the same considerations apply to $\alpha z$, for any $\alpha>0$.
This contradicts the robust no-arbitrage condition.
\end{proof}
\begin{proof}[Proof of Theorem \ref{thm: pennanen}]
 Necessity follows from Lemma \ref{lem: suff Penn}. Sufficiency follows from Lemma \ref{lem: necess Penn}.
\end{proof}

\begin{remark} Our findings are not directly comparable with those of \citep{P11}. First our results are shown without any reference probability measure and second the notion of arbitrage and dual elements are different. In particular, as the existence of a riskless asset $B$ is not assumed in \citep{P11}, the dual elements are martingale deflators as opposed to martingale measures. We also note that the Fundamental Theorem of Asset Pricing in \citep{P11} holds under the additional hypothesis that $S^\infty$ is finite.
In particular, Theorem 5.4 in \citep{P11} cannot be applied to models of superlinear transaction costs as in, e.g. \citep{CR07}.
We do not require this assumption here.
\end{remark}
\begin{remark}
In general, the family of correspondences $C\coloneqq(C_t)_{t\in\cI}$, with $C_t\coloneqq\{x\in\R^d| S_t(\cdot,x)\leq 0\}$, describes a market with physical delivery where the solvency region is convex rather than conical. These models have been studied in \citep{PP10} where the notion of $\NA^r$ is given in terms of recession cones. Namely, $C$ satisfies robust no-arbitrage if and only if $C^\infty$ satisfies $\NA^r$ as in Section \ref{section Kabanov}, where $C^\infty$ is the family of recession cones associated to $C$. It is clear that Theorem \ref{thm: FTAP Kabanov} extends in a straightforward way to this case. Note also that, as opposed to \citep{PP10}, portfolio constraints are also allowed in our model.
\end{remark}
    
\appendix
\section{Some technical tools}

\begin{lemma}\label{lem:appendix sum}
Let $A,\,B\colon\Omega\rightrightarrows\R^d$ be two $\cF_t$ measurable correspondences and let $\zeta\in\cL(\cF_t;A+B)$ be the selection of the sum. Then, we may write $\zeta = \eta + \theta$ with some $\eta\in\cL(\cF_t;A)$ and $\theta\in\cL(\cF_t;B)$.
\end{lemma}
\begin{proof}
By Proposition~14.11(d) in~\citep{R} the correspondence $A\times B\colon\Omega\rightrightarrows\R^{2d}$ is measurable. The correspondence $C(\omega)\coloneqq\{(x,y)\in\R^d\times\R^d\,|\,x+y=\zeta(\omega)\}$ is measurable by Theorem~14.13(a) in~\citep{R}. It is enough to take any selection $(\eta,\theta)\in (A\times B)\cap C$; see Proposition~14.11(a) in~\citep{R}.
\end{proof}
\begin{lemma}\label{lem:selection ri}
	Let $U\colon\Omega\rightrightarrows\R^k$ be an $\cF_{t+1}$ measurable, convex valued correspondence (not necessarily closed). Then it admits a measurable selection taking values in its relative interior.
\end{lemma}
\begin{proof}
	For each $n$ define an $\cF_{t+1}$-measurable, closed, convex valued correspondence $U^n\coloneqq\cl{(U\cap B_n(0))}$, where $B_n(0)$ is the closed ball around $0$ of radius $n$. Define a sequence of sets by $D_0\coloneqq\varnothing$ and $D_{n+1}\coloneqq\dom U^n\backslash D_n$ and let $(\zeta^n_k)_k$ be the Castaing representation of $U^n$. Then
	$$
	\zeta \coloneqq \sum_{n\in\N}\mathbbm{1}_{D_n}\sum_{k\in\N}2^{-k}\zeta^n_k
	$$
	is the sought for selection.
\end{proof}

\bibliographystyle{apalike}

\begin{thebibliography}{99}

\bibitem[Acciaio et al(2016)]{ABPS16} 
Acciaio B., Beiglb\"{o}ck M., Penkner F., \& Schachermayer W. (2016). 
A model-free version of the fundamental theorem of asset pricing and the super-replication theorem.
 \textit{Math. Fin.}, 26(2), 233--251

\bibitem[Aliprantis, Border(2006)]{Aliprantis} 
Aliprantis C. D., \& Border K. C. (2006). 
\textit{Infinite Dimensional Analysis}. 
Berlin: Springer.

\bibitem[Astic, Touzi(2007)]{AT07}
Astic F., \& Touzi N. (2007). 
No arbitrage conditions and liquidity. 
\textit{J. Math. Econ.} 43, 692--708.

\bibitem[Bayraktar, Zhang(2016)]{BZ16}
Bayraktar E., \& Zhang Y. (2016).
Fundamental Theorem of Asset Pricing under transaction costs and model uncertainty.
\textit{Math. Oper. Res.}, 41 (3), 1039--1054.

\bibitem[Bayraktar, Zhou(2017)]{BZ17}
Bayraktar E., \& Zhou Z. (2017).
On Arbitrage and Duality under Model Uncertainty and Portfolio Constraints
\textit{Math. Fin.}, 27 (4) 988--1012. 

\bibitem[Bouchard, Nutz(2015)]{BN15}
Bouchard B., \& Nutz, M. (2015).
Arbitrage and duality in nondominated discrete-time models.
\textit{Ann. App. Prob.}, 2(25), 823--859.

\bibitem[Bouchard, Nutz(2016)]{BN16}
Bouchard B., \& Nutz, M. (2016).
Consistent Price Systems under Model Uncertainty.
\textit{Fin. Stoch.}, 1(20), 83--98.

\bibitem[Bertsekas, Shreve(1978)]{BS}
Bertsekas, D.P., \& Shreve S.E. (1978).
\textit{Stochastic Optimal Control. The discrete time case}. 
New York: Academic Press.

\bibitem[Burzoni et al(2016a)]{BFM16} 
Burzoni, M., Frittelli, M., \& Maggis, M. (2016)
Universal Arbitrage Aggregator in discrete time Markets under Uncertainty.
\textit{Fin. Stoch.}, 1(20), 1--50.

\bibitem[Burzoni(2016)]{Bu16}
Burzoni M. (2016).
Arbitrage and Hedging in Model Independent Markets with frictions.
\textit{SIAM J. Fin. Math.} 7(1), 812--844.

\bibitem[Burzoni et al(2016b)]{BFHMO16}
Burzoni M., Frittelli M., Hou Z., Maggis M., \& Ob\l \'oj J(2016).
Pointwise arbitrage pricing theory in discrete-time.
To appear in \textit{Math. Oper. Res.}, arXiv 1612.07618.

\bibitem[Burzoni et al(2017)]{BRS17}
Burzoni M., Riedel F., \& Soner H.M. (2017) 
Viability and arbitrage under Knightian Uncertainty. 
Preprint: arXiv 1707.03335.

\bibitem[Bartl et al(2017)]{BCKT17}
Bartl D., Cheridito P., Kupper M., \& Tangpi L. (2017).
Duality for increasing convex functionals with countably many marginal constraints.
\textit{Banach J. Math. An.}, 1(11), 72--89.

\bibitem[Cheridito(2016)]{CKT16}
Cheridito P., Kupper M., \& Tangpi L. (2016). 
Duality formulas for robust pricing and hedging in discrete time.
\textit{SIAM J. Fin. Math.}, 8(1), 738--765

\bibitem[\c{C}etin, Rogers(2007)]{CR07}
\c{C}etin, U., \& Rogers, L. C. G. (2007). 
Modeling liquidity effects in discrete time. 
\textit{Math. Fin.} 17(1), 15--29.

\bibitem[\c{C}etin et al(2012)]{CR12}
\c{C}etin U., Jarrow R.A., \& Protter P. (2012).
Liquidity risk and arbitrage pricing theory. 
\textit{Financial Derivatives Pricing}, World Scientific, 153--183.
  
\bibitem[Dellacherie, Meyer(1982)]{DM82}
Dellacherie C., \& Meyer P. (1982). 
\emph{Probabilities and Potential B}. 
Amsterdam New York: North-Holland.

\bibitem[Dolinsky, Soner(2014)]{DS14}
Dolinsky Y., \& Soner H.M. (2014).
Robust hedging with proportional transaction costs.
\textit{Fin. Stoch.}, 18(2), 327--347

\bibitem[Jacod, Shiryaev(1998)]{JS98}
Jacod J., \& Shiryaev A. (1998).
Local martingales and the fundamental asset pricing theorems in the discrete-time case.
\textit{Fin. Stoch.}, 2, 259--273.


\bibitem[Kabanov(1999)]{K99}
Kabanov Y. (1999).
Hedging and liquidation under transaction costs in currency markets. 
\textit{Fin. Stoch.}, 3, 237--248.


\bibitem[Kabanov et al(2002)]{KRS02}
Kabanov Y., R\'asonyi M., \& Stricker C. (2002).
No-arbitrage criteria for financial markets with efficient friction.
\textit{Fin. Stoch.}, 6, 371--382.


\bibitem[Kabanov et al(2003)]{KRS03}
Kabanov Y., R\'asonyi M., \& Stricker C. (2003).
On the closedness of sums of convex cones in $L^0$ and the robust no-arbitrage property.
\textit{Fin. Stoch.}, 7, 403--411.

\bibitem[Molchanov(2005)]{Mol}
 Molchanov I. (2005).
 \emph{Theory of Random Sets}.
 London: Springer.

\bibitem[Pennanen(2011)]{P11}
Pennanen T. (2011).
Arbitrage and deflators in illiquid markets.
\textit{Fin. Stoch.}, 15, 57--83.


\bibitem[Pennanen, Penner(2010)]{PP10}
Pennanen T., \& Penner I. (2010). 
Hedging of claims with physical delivery under convex transaction costs.
\textit{SIAM J. Fin. Math.}, 1, 158--178.


\bibitem[Riedel(2015)]{Ri15}
Riedel F. (2015).
Financial economics without probabilistic prior assumptions.
\textit{Dec. Econ. Fin.}, 38 (1), 75-91.


\bibitem[Rockafellar(1970)]{R70}
Rockafellar T. (1970).
\emph{Convex Analysis}. 
Princeton University Press.

\bibitem[Rockafellar, Wets(2004)]{R}
 Rockafellar T., \& Wets R. (2004).
 \emph{Variational Analysis}.
 Berlin: Springer.

\bibitem[Rokhlin(2007)]{Rok}
Rokhlin D. B. (2007).
Martingale selection problem and asset pricing in finite discrete time.
\textit{Elect. Comm. in Probab.}, 12, 1--8.

\bibitem[Rokhlin(2006)]{Rok06}
Rokhlin D. B. (2006).
A martingale selection problem in the finite discrete-time case.
\textit{Theory Probab. Appl.}, 50(3), 420--435.

\bibitem[Schachermayer(1992)]{Sch94}
Schachermayer, W. (1992).
A Hilbert space proof of the fundamental theorem of asset pricing in finite discrete time.
\textit{Insurance Math. Econom.}, 11(4), 249--257.


\bibitem[Schachermayer(2004)]{S04}
Schachermayer W. (2004).
The Fundamental Theorem of Asset Pricing under proportional transaction costs in finite discrete time.
\textit{Math. Fin.}, 14(1), 19--48.

\end{thebibliography}

\end{document}